\documentclass[11pt]{article}
\usepackage{epsf}
\usepackage{epsfig}
\usepackage{graphicx}
\usepackage{color}
\usepackage{amsmath}
\usepackage{mathrsfs}

\usepackage{amsthm}
\usepackage{latexsym}
\usepackage{amssymb}
\usepackage{amsmath}

\usepackage{stmaryrd}

\usepackage{epsf}
\usepackage{verbatim}

\usepackage[margin=1in]{geometry}
\usepackage{xspace}
\usepackage{xcolor}
\usepackage{setspace}

\newtheorem{coro}{Corollary}[section]

 \newtheorem{remark}{Remark} 
 
\newtheorem{theorem}{Theorem}[section]
\newtheorem{definition}{Definition}[section]
\newtheorem{lemma}{Lemma}[section]

\catcode`@=11 \@addtoreset{equation}{section} 

\begin{document}

\title{{\bf A Collapse Theorem for Holographic Algorithms with Matchgates on
Domain Size At Most 4}}

\vspace{0.3in}
\author{Jin-Yi Cai\thanks{University of Wisconsin-Madison and Peking University.
 {\tt jyc@cs.wisc.edu}. Supported by NSF CCF-0914969 and NSF CCF-1217549.}
\and Zhiguo Fu\thanks{,Department of Computer Science and Engineering, Shanghai Jiao Tong University. {\tt fuzg@jlu.edu.cn}}}

\date{}
\maketitle
\vspace{0.3in}
\begin{abstract}
Holographic algorithms with matchgates are a novel approach
to design polynomial time computation.
It uses Kasteleyn's algorithm
for perfect matchings, and more importantly  a holographic reduction.
The two fundamental parameters
of a holographic reduction are the domain size $k$
of the underlying
problem, and the basis size $\ell$. A holographic reduction
transforms the computation to matchgates
by a linear transformation  that maps
to (a tensor product space of) a linear space of dimension $2^\ell$.
We prove a sharp basis collapse theorem, that shows that
for domain size 3 and 4,
all non-trivial holographic reductions
have basis size $\ell$ collapse to 1 and 2 respectively.
The main proof techniques are Matchgate Identities,
and a Group Property of matchgate signatures.
\newline

\end{abstract}

\thispagestyle{empty}
\newpage
\setcounter{page}{1}
\section{Introduction}
Matchgates were first introduced by Leslie Valiant~\cite{string22}
to show that a non-trivial, though
restricted, fragment of
quantum computation can be simulated in classical polynomial time.
Subsequently he introduced holographic algorithms with matchgates
\cite{string23} as a methodology to design
 polynomial time algorithms for some problems which seem to require exponential time. Computation in these algorithms is expressed and interpreted through a choice of linear basis vectors in an exponential ``holographic" mix. Then the actual computation is carried out, via the Holant Theorem, by Kasteleyn's algorithm (a.k.a. the  FKT
algorithm) \cite{string17,string18,string20} for counting the number of perfect matchings in a planar graph. This methodology has produced polynomial time algorithms for a number of problems, and minor variations of which are known to be NP-hard.
The results are often surprising and counter-intuitive.
For example, it is shown~\cite{string24} that a
restrictive SAT problem  $\#_7$Pl-Rtw-Mon-3CNF
(counting the number of satisfying
assignments of a planar read-twice monotone 3CNF formula, modulo 7)
is solvable in polynomial time.
The same counting problem without mod 7 is known to be
\#P-complete, thus fully general despite its apparent syntactic
restriction; the problem mod 2 is $\oplus $P-complete, and thus
NP-hard under randomized reductions. And yet, the problem mod 7 is
tractable. Such ``anomaly'' challenges
our conception of what polynomial time computation can do,
and where the frontier between P and \#P lies, assuming
they are truly different.

These holographic algorithms are quite exotic,
and use a quantum-like superposition of fragments of
computation to achieve a pattern of interference and cancellations.
Since we lack any good absolute lower
bounds that apply to unrestricted computational models,
we should ask ourselves why do we believe
those conjectures
such as P $\not =$ NP or P $\not =$ ${\rm P}^{\rm \#P}$
that are at the foundation of our discipline.
We posit that the only defensible argument is
the observed inability of existing algorithmic techniques to solve
NP-hard or \#P-hard problems in polynomial time.
Now these new holographic algorithms are quite unlike
the existing algorithmic techniques, and thus pose a new challenge.
To maintain the credibility of these widely believed conjectures,
and the self-respect of the discipline,
we must gain a better understanding of
 what the new methodology
can or cannot do.
To quote Valiant~\cite{string23},
``The most intriguing question, clearly,
is whether polynomial time holographic algorithms exist
for NP- or \#P-complete problems.
$\ldots$ [T]he existence of
such a reduction would be implied by the solvability of a
finite system of polynomial equations $\ldots$
[A]ny proof of P $\not =$ NP
may need to
explain, and not only to imply, the unsolvability of our
polynomial systems.''

Substantial progress has been made. For example,
the appearance of the modulus 7 for 3CNF, which was considered peculiar,
has been ``explained'' by
the fact that $7 = 2^3 -1$ is a Mersenne prime~\cite{string7}.
%
Thus, e.g., $\#_{31}$Pl-Rtw-Mon-5CNF is in P by
the same holographic algorithm.
Such understanding is achieved only after a systematic
study of the {\it structural theory} of holographic algorithms.
This is the theory of what holographic reductions are possible,
and what they can do
with matchgates.

In the design of a holographic algorithm, a crucial step is a choice of
linear basis vectors, through which the computation is expressed and interpreted.
Because the underlying basic computation is ultimately reduced to perfect matchings, the linear basis vectors are of dimension $2^{\ell}$, where $\ell$ is called the size of the basis. For a general CSP-type counting problem, one can assume there is a natural parameter $k$, called its domain size. This is the range over which variables take values. For example, Boolean CSP problems all have domain size 2. A $k$-coloring problem on graphs has domain size $k$. In holographic algorithms of domain size $k$, the linear basis has $k$ vectors of dimension $2^{\ell}$, which can be expressed as a $2^{\ell}\times k$ matrix.

Utilizing bases of an arbitrarily  large but fixed size $\ell$  is
 a theoretical possibility which may allow for an unlimited
variety of holographic algorithms.
For example, the algorithm for $\#_7$Pl-Rtw-Mon-3CNF in \cite{string24}
originally used a basis of size $\ell=2$, expressed as a $4 \times 2$  matrix.
However, over the Boolean domain ($k=2$),
Cai and Lu~\cite{string6}
proved  a surprising collapse theorem that
{\it any} non-trivial holographic algorithm on a basis of size $\ell\geq 2$
can be simulated
on a basis of size 1. (In particular, for $\#_7$Pl-Rtw-Mon-3CNF, there is a
linear basis with $2$ vectors of dimension $2^{1}$,
expressed as a  $2 \times 2$  matrix.)
  This is the fundamental rationale
to develop the theory for Boolean domain holographic algorithms
over the group {\bf GL}$_2(\mathbb{C})$~\cite{string7},
which is the foundation for all the systematic results that
have been achieved.

While this drastic collapse from an arbitrary $2^\ell$ to $2$ is surprising,
perhaps there is a plausible philosophical justification.
One might reason that for the Boolean domain $k=2$, ``information theoretically''
one should need only 2 dimensions to encode data.
(This is by no means a proof! It is technically false, as most
philosophical arguments are, since provably there are holographic
reductions in
{\bf GL}$_2(\mathbb{C})$ that cannot be done in {\bf GL}$_2(\mathbb{R})$.)
Nevertheless, following this logic,
an audacious but plausible conjecture is that
for a general domain size $k$, there is a collapse
to the smallest $\ell$ such that $2^\ell \ge k$.

In this paper we prove a basis
collapse theorem for holographic algorithms on domain size 3 and 4.
For domain size 3, we show that all non-trivial holographic algorithms
with matchgates using a basis of size $\ell$ can be simulated by
a basis of size $1$.
For domain size 4, we show that it collapses to  size 2.
Thus, for domain size 4, the proper transformation theory should be
developed over the group {\bf GL}$_4(\mathbb{C})$.
Note that there is a further surprise that the collapse
for domain size 3  is not to dimension 4, but to dimension 2.
This turns out to be a consequence of some very delicate
properties of matchgates, philosophical arguments not withstanding.
In \cite{string25}, Valiant
 gave holographic algorithms for several
 interesting problems on domain size 3.  Holographic algorithms for
domain size 4 or above are largely unexplored. The results
of this paper are the first steps toward this goal.
It shows that for domain size 4 we should develop
the  theory on {\bf GL}$_4(\mathbb{C})$,
rather than on an infinite set of dimensions.

Our main proof techniques are Matchgate Identities,
and a Group Property of matchgate signatures.
A matchgate is a planar graph associated
with a function called its (standard) signature, which
represents its perfect matching properties.
Matchgate Identities are a set of necessary and
sufficient conditions for a matchgate (standard)  signature.
In \cite{string6}, the proof of the collapse theorem on
domain size 2 heavily depends on intricate constructions of matchgates;
but this is difficult to generalize to domain size $k>2$.
Instead we introduce a new algebraic proof technique
that heavily depends on Matchgate Identities.
We will ``construct'' the required combinatorial objects---matchgates---by
purely algebraic means. Starting from certain
presumed holographic reductions, we will combine together
algebraic objects, which we prove that they must
correspond to matchgates.
The most difficult step is to extract out a rank 4 submatrix
of a certain signature matrix, using Matchgate Identities.
In one crucial step we also use a Group Property
that matchgate signatures satisfy, and use the algebraic
inverse to obtain the combinatorial object.
This indirect way of construction is similar
to the way G\"{o}del's completeness theorem is proved
(as simplified by Leon Henkin),
where one builds a semantic object---a model---out of given
syntactic objects, namely  a consistent set of formulae.
Our process is the reverse: We build concrete syntactic
objects (matchgates) out of presumed semantic linear transformations.

This theory fits in a  broader picture.
Over the past few years a string of complexity dichotomy
theorems have been proved~\cite{Bulatov,cailuxia10,Dyer-Richerby1,
Dyer-Goldberg-Jerrum,GGJT,GH,Cai-Chen,
Cai-Kowalczyk-complex,Cai-Kowalczyk-Williams,caiguowilliams13,Guo-Lu-Valiant,
Guo-Williams}
which support parts of the following
overall thesis:  For a wide class of counting problems
expressible as partition functions
defined by local constraints, or sum-of-product computations,
{\it every single problem} can be classified into one of three types.
The first type is called tractable problems, which can be
solved in polynomial time over arbitrary structures.
The second type consists of problems that are
 \#P-hard over general structures, but solvable in
polynomial time over planar structures.
The third type problems are those which remain \#P-hard
over planar structures.
Moreover, the second type of problems are precisely those which
are solvable by a holographic reduction to matchgates.
Thus, the new methodology of holographic algorithms
with matchgates constitutes a {\it universal}
algorithm for all such problems.
It is possible that the ultimate significance of this new
methodology introduced by Valiant~\cite{string23} lies in
its pivotal r\^{o}le in this classification program.
We note that for decades researchers in physics have studied
the so-called ``Exactly Solvable Models'' (see e.g., \cite{Baxter}).
The classification program,
especially the universality part
about the new holographic algorithms
with matchgates, if true,
would provide an answer from computer science.

However, the provable part of this conjectured
universality of holographic algorithms
with matchgates is essentially restricted to the Boolean domain.
The full scope of this thesis is beyond our ability to prove now.
A main obstacle is that the theory of holographic algorithms
with matchgates has not been adequately developed for domain size
greater than 2. This paper is  a necessary first step in that program.

 This paper is organized as follows. In Section 2, we briefly give
the background and some notations. In Section 3, we
introduce degenerate and full rank signatures.
In Section 4, we give a new algebraic
proof for the collapse theorem on domain size 2, whereby introducing
the new technique in a simpler setting.
In Section 5, we give the collapse theorems on domain size 3 and 4.
An illustrative problem solved by a holographic algorithm
using matchgates on domain size 4 is
given in the appendix.
\section{Background and Notations}
\subsection{Background}
In this section,  we review some definitions and results.  More details can be found in \cite{string23,string7}.

A matchgate $\Gamma$ is a triple $(G, X, Y)$ where $G$ is a planar embedding of a planar graph $(V, E, W)$ where $X \subseteq V$ is a set of input nodes and $Y \subseteq V$ is a set of output nodes, and where $X, Y$ are disjoint. Further, as one proceeds counterclockwise around the outer face starting from one point one encounters first the input nodes labeled $1, 2, \cdots, |X|$ and then the output nodes $|Y|, \cdots, 2, 1$ in that order. The arity of the matchgate is $|X|+|Y|$. For $Z\subseteq X\bigcup Y$ we define the standard signature of $\Gamma$ with respect to $Z$ to be PerfMatch$(G-Z)$, where $G-Z$ is the graph obtained by removing from $G$ the node set $Z$ and all edges that are incident to $Z$,
 and PerfMatch$(G-Z)$ is the sum, over all perfect matchings $M$ of $G-Z$, of the product of the weights of matching edges in $M$.
Note that when all edges have weight 1, then PerfMatch$(G)$ counts
the number of perfect matchings.
We define the standard signature of $\Gamma$ to be the following
 $2^{|Y|}\times 2^{|X|}$ matrix
$\underline{\Gamma}$
 row-indexed by output nodes and column-indexed by input nodes (note that in $\cite{string23}$, $\underline{\Gamma}$ is a $2^{|X|}\times 2^{|Y|}$ matrix
 row-indexed by input nodes and column-indexed by output nodes).
The entries of $\underline{\Gamma}$ are standard signatures  of $\Gamma$ with respect to $Z$ for the $2^{|X|+|Y|}$ choices of $Z$. The labeling of the matrix
is as follows: Suppose that $X$ and $Y$ have the labeling described, i.e., the nodes are labeled $1, 2, \cdots, |X|$ and $|Y|, \cdots, 2, 1$ in counterclockwise order. Then each choice  of $Z$ is a subset of  $X\bigcup Y$. If each node present in $Z$ is denoted by a bit 1, and each node absent by a bit 0, then we have two binary strings in $\{0, 1\}^{|X|}$
and $\{0, 1\}^{|Y|}$ respectively,
where the nodes labeled 1 (for both  $X$ and $Y$)
correspond to the leftmost binary bit. Suppose that $i, j$ are the numbers represented by these strings in binary. Then the entry corresponding to $Z$ will be the one in row $i$ and column $j$ in the signature matrix $\underline{\Gamma}$.
This label ordering will allow the planar composition of matchgates
connecting input nodes of
one with the output nodes of another nicely correspond
to matrix product.

A matchgate $\Gamma$ is an odd (resp.  even) matchgate if it has an odd
(resp.  even) number of nodes.

 Let $\Gamma$ be a matchgate. If $\Gamma$ has no input nodes, then it is called a generator matchgate. If $\Gamma$ has no output nodes, then it is called a recognizer matchgate. Otherwise $\Gamma$ is called a transducer matchgate.
Note that the standard signature $\underline{G}$ of a generator matchgate is a column vector and the standard signature $\underline{R}$ of a recognizer matchgate is a row vector.

From the definition of standard signatures, we directly have the following Lemma.
\begin{lemma}
\label{product transformer}
Let $\underline{R}$ be the standard signature of a recognizer matchgate of arity $n\ell$ and $T$ be the standard signature of a transducer matchgate of $\ell$-output and $s$-input, then $\underline{R}'=\underline{R}T^{\otimes n}$ is the standard signature of a recognizer matchgate of arity $ns$.
\end{lemma}

On the other hand,
we can view the standard signature of an $n$-output generator matchgate as a contravariant  tensor $\textbf{G}$ with $n$ (upper) indices.
Under the standard basis $[e_{0}~e_{1}]
=
\begin{pmatrix}
1& 0\\
0& 1
\end{pmatrix}$,
it takes the form $\underline{G}$
with $2^{n}$ entries,  where
\[\underline{G}^{i_{1}i_{2}\cdots i_{n}}= {\rm PerfMatch}(G-Z),  ~~~~
i_{1}, i_{2}, \ldots,  i_{n} \in \{0,  1\}. \]
Here $Z$ is the subset of the output nodes having the characteristic
sequence $\chi_{Z}=i_{1}i_{2}\cdots i_{n}$,
in which $i_j$ is the bit for the output node labeled $j$, and
$G-Z$ is the graph obtained from $G$ by removing $Z$ and its incident edges.
Then the column vector $\underline{G}=(\underline{G}^{i_{1}i_{2}\cdots i_{n}})$ whose entries are ordered lexicographically according to $\chi_{Z}$
is the standard signature of a generator matchgate.

Similarly a recognizer matchgate with $n$ input nodes is assigned a covariant tensor $\textbf{R}$ with $n$ (lower) indices.
Under the standard basis $[e_{0}~e_{1}]$,  it takes the form $\underline{R}$ with $2^{n}$ entries,
\[
\underline{R}_{i_{1}i_{2}\cdots i_{n}}= {\rm PerfMatch}(G-Z), ~~~~i_{1},
i_{2}, \ldots,  i_{n} \in \{0,  1\},\]
where $Z$ is the subset of the input nodes having the characteristic sequence $\chi_{Z}=i_{1}i_{2}\cdots i_{n}$.
Then the row vector $\underline{R}=(\underline{R}_{i_{1}i_{2}\cdots i_{n}})$ whose entries are ordered lexicographically according to $\chi_{Z}$ is the standard signature of a recognizer matchgate.

 A $basis~M=(m_{1}  , m_{2} , \cdots , m_{k})$  contains $k$ vectors, each of them has dimension $2^{\ell}$ (size $\ell$).
 We use the following notation: $M=(a_{i}^{\alpha})$, where lower index $i\in [k]$ is for column and upper index $\alpha\in\{0, 1\}^{\ell}$ is for row.
A basis $M$ need not be linearly independent.
We say $M$ has full rank if rank$(M)=k$. In the present paper,
we assume that $M$ has full rank  (thus $2^{\ell}\geq k$)
unless otherwise specified.

Under a basis $M$,  we can talk about the signature of a matchgate
after the transformation.


\begin{definition}
The contravariant tensor $\textbf{G}$ of a generator matchgate $\Gamma$
of arity $n$ has signature $G$
(written as a column vector)
 under basis $M$ iff $M^{\otimes n}G=\underline{G}$ is the standard signature of the generator matchgate $\Gamma$.
\end{definition}


\begin{definition}
The covariant tensor $\textbf{R}$ of a recognizer matchgate $\Gamma'$
of arity $n$  has signature $R$
(written as a row vector)
 under basis $M$ iff $\underline{R}M^{\otimes n}=R$ where $\underline{R}$ is the standard signature of the recognizer matchgate $\Gamma'$.
\end{definition}


\begin{definition}
A contravariant tensor $\textbf{G}$
(resp.  a covariant tensor $\textbf{R}$)
is realizable over a basis $M$ iff there exists a generator matchgate $\Gamma$ (resp.  a recognizer matchgate $\Gamma'$) such that $G$ (resp.  $R$) is the signature of $\Gamma$ (resp.  $\Gamma'$) under basis $M$.
They are simultaneously realizable if they are realizable over a
common basis.
\end{definition}

\begin{remark}
Under a basis of size $\ell$, if a general signature has arity $n$, then the standard signature is of arity $n\ell$, where
 $n\ell$ is the number of external
nodes in the matchgate. So a standard  generator signature $\underline{G}$ (resp. a standard recognizer signature $\underline{R}$) has $2^{n\ell}$ entries.
 We use $\underline{G}^{\alpha_{1}\alpha_{2}\cdots\alpha_{n}}$, where each $\alpha_{i}\in \{0, 1\}^{\ell}$, to denote the blockwise form of the signature
entry of $\underline{G}$ of arity $n\ell$.
Similarly we use the notation
$\underline{R}_{\alpha_{1}\alpha_{2}\cdots \alpha_{n}}$
 for a recognizer signature.
\end{remark}
Then we have
\begin{center}
$\underline{G}^{\alpha_{1}\alpha_{2}\cdots \alpha_{n}}=\displaystyle
 \sum_{j_{1}, j_{2}, \ldots,  j_{n}\in [k]}G^{j_{1}j_{2}\cdots j_{n}}a_{j_{1}}^{\alpha_{1}}a_{j_{2}}^{\alpha_{2}}\cdots a_{j_{n}}^{\alpha_{n}}$,
\end{center}
where $\alpha_{i}\in\{0, 1\}^{\ell}$,  for $i=1, 2, \cdots, n$.


\begin{center}
$R_{j_{1}j_{2}\cdots j_{n}}=\displaystyle \sum_{\alpha_{1}, \alpha_{2}, \ldots, \alpha_{n}\in \{0, 1\}^{\ell}}\underline{R}_{\alpha_{1}\alpha_{2}\cdots \alpha_{n}}a_{j_{1}}^{\alpha_{1}}a_{j_{2}}^{\alpha_{2}}\cdots a_{j_{n}}^{\alpha_{n}}$,
\end{center}
where $j_{i}\in[k]$,
 for $i=1, 2, \cdots, n$.

A matchgrid $\Omega=(A, B, C)$ is a weighted planar graph consisting of a disjoint union of: a set of (not necessarily distinct) $g$ generator matchgates  $A=\{A_{1}, A_{2}, \cdots, A_{g}\}$,  a set of  (not necessarily distinct)
 $r$ recognizer matchgates $B=\{B_{1}, B_{2}, \cdots, B_{r}\}$,  and a set of $f$ connecting edges $C=\{C_{1}, C_{2}, \cdots, C_{f}\}$,  where each $C_{i}$ edge has weight 1 and joins an output node of a generator matchgate with an input node of a recognizer matchgate,  so that every input and output node in every constituent matchgate has exactly one such incident connecting edge.


Let $G(A_{i}, M)$ be the signature of generator matchgate $A_{i}$ under the basis $M$ and $R(B_{j}, M)$ be the signature of recognizer matchgate $B_{j}$ under the basis $M$.
Let $G=\bigotimes_{i=1}^{g}G(A_{i}, M)$ and $R=\bigotimes_{j=1}^{r}R(B_{j}, M)$
be their tensor product,
then $\rm{Holant}(\Omega)$ is defined to be the {\it contraction}
of these two product tensors
(the sum over all indices of the product of the
corresponding values of $G$ and $R$),
where the corresponding indices match up according to the $f$ connecting edges in $C$.


Valiant's {Holant} Theorem is

\begin{theorem}(Valiant $\cite{string23}$)
For any mathcgrid $\Omega$ over any basis $M$,  let $\Gamma$ be its underlying weighted graph,
 then
\begin{center}
${\rm{Holant}}(\Omega)={\rm{PerfMatch}}(\Gamma)$.
\end{center}
\end{theorem}


The FKT algorithm can compute the weighted sum of
 perfect matchings $\rm{PerfMatch}(\Gamma)$   for a planar graph in polynomial time.  So $\rm{Holant}(\Omega)$ is computable in polynomial time.

\subsection{Matrix Form of Signatures}

\begin{definition}\label{t-matrixform-G}
For a generator signature $G=(G^{j_{1}j_{2}\cdots j_{n}})$ on domain size $k$,
the $t$-th matrix form $G(t)$ ($1 \le t \le n$)
 is a $k\times k^{n-1}$ matrix, where the rows are
indexed by $1 \le j_{t} \le t$ and the columns are
indexed by $j_{1}\cdots j_{t-1}j_{t+1}\cdots j_{n}$ in lexicographic order.
\end{definition}
\begin{definition}\label{t-matrixform-R}
For a recognizer signature $R=(R_{j_{1}j_{2}\cdots j_{n}})$ on domain size $k$,
the $t$-th matrix form
$R(t)$ ($1 \le t \le n$)
 is a $k^{n-1}\times k$ matrix where
the rows are indexed by $j_{1}\cdots j_{t-1}j_{t+1}\cdots j_{n}$
in lexicographic order
and the columns are indexed by $1 \le j_{t} \le k$.
\end{definition}
For example, let  $G=(G^{j_{1}j_{2}})$ and $R=(R_{j_{1}j_{2}})$ where $n=2$ and $k=3$, then
\begin{center}
$G(1)=\begin{pmatrix}
G^{11}&G^{12}&G^{13}\\
G^{21}&G^{22}&G^{23}\\
G^{31}&G^{32}&G^{33}\end{pmatrix}$,
~~~~
$R(1)=\begin{pmatrix}
R_{11}&R_{21}&R_{31}\\
R_{12}&R_{22}&R_{32}\\
R_{13}&R_{23}&R_{33}
\end{pmatrix}$,
\end{center}

We may consider a standard signature of arity $n \ell$
as a signature on domain size $k=2^\ell$, with the identity
matrix $I_{2^\ell}$, then the following are
special cases of Definition~\ref{t-matrixform-G} and \ref{t-matrixform-R}.

The $t$-th matrix form $\underline{G}(t)$ ($1 \le t \le n$)
of the standard signature $\underline{G}=(\underline{G}^{\alpha_{1}\alpha_{2}\cdots \alpha_{n}})$ of a generator matchgate  of arity $n\ell$
 is a $2^{\ell}\times 2^{(n-1)\ell}$ matrix. Its rows are
indexed by  $\alpha_{t}$ and its columns are indexed
by $\alpha_{1}\cdots \alpha_{t-1}\alpha_{t+1}\cdots \alpha_{n}$.
%
The $t$-th matrix form $\underline{R}(t)$ ($1 \le t \le n$)
of the standard signature  $\underline{R}=(\underline{R}_{\alpha_{1}\alpha_{2}\cdots \alpha_{n}})$ of a recognizer matchgate of arity $n\ell$
 is a $2^{(n-1)\ell}\times 2^{\ell}$  matrix.  Its rows are
 indexed by $\alpha_{1}\cdots \alpha_{t-1}\alpha_{t+1}\cdots \alpha_{n}$ and
its columns are indexed by $\alpha_{t}$.
For example, let $\underline{G}=(\underline{G}^{\alpha_{1}\alpha_{2}})$, $\underline{R}=(\underline{R}_{\alpha_{1}\alpha_{2}})$ where $n=2, \ell=2$, then
\begin{center}
$\underline{G}(1)=\begin{pmatrix}
\underline{G}^{0000}& \underline{G}^{0001}& \underline{G}^{0010}& \underline{G}^{0011}\\
\underline{G}^{0100}& \underline{G}^{0101}& \underline{G}^{0110}& \underline{G}^{0111}\\
\underline{G}^{1000}& \underline{G}^{1001}& \underline{G}^{1010}& \underline{G}^{1011}\\
\underline{G}^{1100}& \underline{G}^{1101}& \underline{G}^{1110}& \underline{G}^{1111}
\end{pmatrix}$
,
~~~~
$\underline{R}(1)=\begin{pmatrix}
\underline{R}_{0000}& \underline{R}_{0100}& \underline{R}_{1000}& \underline{R}_{1100}\\
\underline{R}_{0001}& \underline{R}_{0101}& \underline{R}_{1001}& \underline{R}_{1101}\\
\underline{R}_{0010}& \underline{R}_{0110}& \underline{R}_{1010}& \underline{R}_{1110}\\
\underline{R}_{0011}& \underline{R}_{0111}& \underline{R}_{1011}& \underline{R}_{1111}
\end{pmatrix}$.
\end{center}

The following lemma can be proved directly.
\begin{lemma}\label{from vec to mat}
If $\underline{G}=M^{\otimes n}G$, where $M$ is a $2^{\ell}\times k$ basis, $G$ is a generator signature of dimension $k^{n}$, then
\begin{equation*}
\underline{G}(t)=MG(t)(M^{\tt{T}})^{\otimes (n-1)}.
\end{equation*}
\end{lemma}

 Denote the Hamming weight of a binary string $\alpha$ as ${\rm wt}(\alpha)$.
%
We denote the row of $\underline{G}(t)$ indexed by $\alpha\in\{0,1\}^{\ell}$
as $\underline{G}(t)^{\alpha}$. The parity of ${\rm wt}(\alpha)$ is
also called the parity of $\underline{G}(t)^{\alpha}$.
The $n\times n$ identity matrix is denoted as $I_{n}$. The transpose of the matrix $A$ is denoted as $A^{\tt{T}}$.

\subsection{Matchgate Identities}

 Let $\underline{G}$ be the standard signature of a matchgate of arity $n$ (we
discuss $\underline{G}$ here, it is the same for $\underline{R}$). A pattern $\alpha$
 is an $n$-bit string, i.e., $\alpha\in\{0, 1\}^{n}$. A position vector
$P=\{p_{1}, \ldots, p_s\}$ is
 a subsequence of $\{1, 2, \ldots, n\}$,
where $p_{i}\in [n]$ and $p_{1}<p_{2}<\cdots<p_{s}$.
A position vector $P$
also denotes the pattern $p$ whose
 $(p_{1}, p_{2}, \ldots, p_{s})$-th bits are 1 and others are 0.
  Let $e_{i}\in \{0, 1\}^{n}$ be the pattern with 1 in the $i$-th bit and 0 elsewhere.
 Let $\alpha\oplus\beta$ be the bitwise XOR  pattern of $\alpha$ and $\beta$. Then for any pattern $\alpha\in \{0, 1\}^{n}$ and
 any position vector $P=\{p_{1}, \ldots, p_s\}$, we have the following Matchgate Identities (MGI):
 \begin{equation}\label{definition of MGI}
 \displaystyle\sum_{i=1}^{s}(-1)^{i}\underline{G}^{\alpha\oplus e_{p_{i}}}\underline{G}^{\alpha\oplus p\oplus e_{p_{i}}}=0.
 \end{equation}

 By the definition of standard signatures, if $\underline{G}$ is the standard signature of an odd matchgate, then $\underline{G}^{\alpha}=0$ for even ${\rm wt}(\alpha)$. If $\underline{G}$ is the standard signature of an even matchgate, then $\underline{G}^{\alpha}=0$ for odd ${\rm wt}(\alpha)$.  This is
the Parity Condition of standard signatures.

 \begin{theorem}\label{MGI}
 A vector in $\mathbb{C}^{2^{n}}$
is the standard signature of a  matchgate iff it satisfies:
 \begin{itemize}
\item the Parity Condition,
\item the Matchgate Identities.
\end{itemize}
 \end{theorem}

For a proof see \cite{Cai-Aaron}.
Actually in \cite{Cai-Aaron} it is shown that
MGI implies the Parity Condition. But in practice, it is easier
to apply the Parity Condition first.
  From the definition of MGI, we directly have the following lemma.
 \begin{lemma}\label{lemma2.2}
 In $\mathrm{MGI}$,  the XOR of the indices in every product term $\underline{G}^{\alpha\oplus e_{i}}\underline{G}^{\alpha\oplus p\oplus e_{i}}$ is the pattern $p$.
%
Assume $\underline{G}$ satisfies the Parity Condition.
If ${\rm wt}(p)$ is odd,
or if  ${\rm wt}(p)=2$,  then the $\mathrm{MGI}$ are automatically satisfied.
\end{lemma}
%
%
%
\begin{proof}
The first statement is obvious.
Hence if ${\rm wt}(p)$ is odd, then every product term is zero
by the Parity Condition.
If ${\rm wt}(p)=2$,  then for
 any pattern $\alpha$, the matchgate identity is
$\underline{G}^{\alpha\oplus e_{p_{1}}}\underline{G}^{\alpha\oplus e_{p_{2}}}-\underline{G}^{\alpha\oplus e_{p_{2}}}\underline{G}^{\alpha\oplus e_{p_{1}}}=0$,
so it is automatically satisfied.
\end{proof}

\begin{lemma}\label{lemma5.4}
A signature $G=(G^{i_{1}i_{2}i_{3}i_{4}})$ of arity 4
is the standard signature of an even matchgate (generator or recognizer) iff $G^{\alpha}=0$ for all odd ${\rm wt}(\alpha)$ and
\begin{equation}\label{evenMGI4}
G^{0000}G^{1111}-G^{1100}G^{0011}+G^{1010}G^{0101}-G^{1001}G^{0110}=0.
\end{equation}
Similarly, it is the standard signature of an odd matchgate (generator or recognizer) iff $G^{\alpha}=0$ for all even ${\rm wt}(\alpha)$ and
\begin{equation}\label{oddMGI4}
G^{1000}G^{0111}-G^{0100}G^{1011}+G^{0010}G^{1101}-G^{0001}G^{1110}=0.
\end{equation}
\end{lemma}
\begin{proof}
We prove the Lemma for the even case. The proof for the odd case is similar and we omit it.

If $G=(G^{i_{1}i_{2}i_{3}i_{4}})$
is the standard signature of an even matchgate, then $G^{\alpha}=0$ for odd ${\rm wt}(\alpha)$ by the Parity Condition.
Equation (\ref{evenMGI4}) follows from MGI where
we choose the position vector $P=\{1,2,3,4\}$ and the pattern $1000$.


Conversely, $G$ satisfies the Parity Condition  as given.
For arity 4, Lemma~\ref{lemma2.2} shows that
the only non-trivial position vector is $P=\{1,2,3,4\}$.
It can be verified  that
for $P=\{1,2,3,4\}$ and any pattern, the MGI is equivalent to
equation (\ref{evenMGI4})
for an even matchgate of arity 4.
Hence $G$ satisfies  MGI, and thus it is a
standard signature by Theorem \ref{MGI}.
\end{proof}

\section{Degenerate and Full Rank Signatures}
\begin{definition}\label{def3.1}
A signature $G$ (generator or recognizer) on domain size $k$ is degenerate iff $G$ has the following form:
\begin{center}
$G=\gamma_{1}\otimes \gamma_{2}\otimes\cdots\otimes \gamma_{n}$,
\end{center}
where $\gamma_{i}$ are vectors of dimension $k$.
\end{definition}

\begin{lemma}
A signature $G$ on domain size $k$  is degenerate iff {\rm rank}$(G(t))\leq 1$ for $1\leq t\leq n$.
\end{lemma}
\begin{proof}
We prove the lemma for generator signatures. The proof for recognizer signatures is similar.

Let $G=(G^{i_{1}i_{2}\cdots i_{n}})$, where $i_{t}\in [k]$. If there exists $t\in[k]$ such that rank$(G(t))=0$, then $G$ is identically zero and is degenerate obviously. If the matrix form $G(t)$ has rank 1 for $1\leq t\leq n$, we will prove the lemma by induction on the arity $n$.

For $n=2$, let $G(1)^{i}$ be the $i$-th row of the first matrix form $G(1)$, then there exists a non-zero row $G(1)^{j}$, where $j\in [k]$,
and $a_{1}, a_{2}, \cdots, a_{k} \in \mathbb{C}$
such that $G(1)^{i}=a_{i}G(1)^{j}$ for $1\leq i\leq k$ from rank$(G(1))=1$.
Let $\gamma_{1}=(a_{1}, a_{2}, \cdots, a_{k}), \gamma_{2}=G(1)^{j}$, then
$G=\gamma_{1}\otimes \gamma_{2}$.

Inductively we assume that the lemma has  been proved for arity
less than $n$.
There exists a non-zero row $G(1)^{j}$ and
$a_{1}, a_{2}, \cdots, a_{k} \in \mathbb{C}$ such that
$G(1)^{i}=a_{i}G(1)^{j}$ for $1\leq i\leq k$ from rank$(G(1))=1$.
Note that $G(1)^{j}$ is a signature of arity $n-1$ and the matrix forms of $G(1)^{j}$ are
  sub-matrices of the matrix forms of $G$,
 thus there exist vectors $\gamma_{2}, \gamma_{3}, \cdots, \gamma_{n}$ of dimension $k$ such that $G(1)^{j}=\gamma_{2}\otimes \gamma_{3}\otimes\cdots\otimes \gamma_{n}$ by induction.
 Let $\gamma_{1}=(a_{1}, a_{2}, \cdots, a_{k})$, then
$G=\gamma_{1}\otimes \gamma_{2}\otimes\cdots\otimes \gamma_{n}$.

Conversely, if $G=\gamma_{1}\otimes \gamma_{2}\otimes\cdots\otimes \gamma_{n}$, it is obvious that the matrix form $G(t)$ has rank at most 1 for $1\leq t\leq n$.
\end{proof}

\begin{definition}\label{def3.2}
For a non-degenerate signature $G$ (generator or recognizer) on domain size $k$, if there exists $t$ such that rank$(G(t))=k$, then $G$ is called a  signature of full rank.
\end{definition}

\begin{remark}
We will prove a collapse theorem for holographic algorithms which employs at least one generator signature of full rank.
For example, we will prove:
For domain size 3, any non-trivial holographic algorithm on a basis of size $\ell\geq 2$ which employs at least one generator signature of full rank can be simulated
on a basis of size 1.
Any holographic algorithm using only degenerate generator signatures
is trivial.
\end{remark}

\section{A New Proof for the Collapse Theorem on Domain Size 2}
The following is a simple lemma from Linear Algebra.
\begin{lemma}\label{lemma 4.1}
Let $A, B, C$ be $m\times n$, $n\times s$, $s\times t$ matrices respectively, where rank$(A)=n$, rank$(C)=s$,
then
\begin{equation*}
{\rm rank}(AB)={\rm rank}(B),~~~~ {\rm rank}(BC)={\rm rank}(B).
\end{equation*}
\end{lemma}



We need to introduce a new notation for
a splicing operation. Let $\alpha \in\{0,1\}^{\ell}$ and
$\beta \in \{0,1\}^{(n-1)\ell}$. Then we
  use $\beta\curvearrowleft_{t}\alpha$ to denote the binary string 
$\alpha_{1}\alpha_{2}\cdots\alpha_{n}\in\{0,1\}^{n\ell}$, where
for each $i$, $\alpha_i \in \{0,1\}^{\ell}$,  $\alpha_t=\alpha$  and
$\alpha_{1}\cdots\alpha_{t-1}\alpha_{t+1}\cdots\alpha_{n} = \beta$.
 Similarly, we denote a position vector $P$ as $(q_{1}q_{2}\cdots q_{d'})\curvearrowleft_{t}(p_{1}p_{2}\cdots p_{d})$, where $p_{i}$ is in the $t$-th block for $1\leq i\leq d$, and $q_{j}$ is in other blocks for $1\leq j\leq d'$.

In this section, assume that the basis $M$ is a $2^{\ell}\times 2$ matrix of rank 2, $G$ is a generator signature of arity $n$
and full rank on domain size 2 (thus there exists $t\in[n]$ such that rank$(G(t))=2$),  and $\underline{G}=M^{\otimes n}G$ is a standard signature of arity $n\ell$.
Note that if $M$ has rank at most one, or if all generators used by
a holographic algorithm are not of full rank (on domain size 2 this
means they are all  degenerate), then the Holant is trivial to compute
and this is a trivial holographic algorithm.

  From Lemma \ref{from vec to mat}, we have $\underline{G}(t)=MG(t)(M^{\tt{T}})^{\otimes (n-1)}$. Then by Lemma \ref{lemma 4.1}, rank$(\underline{G}(t))=2$. 
Therefore we can define $\sigma, \tau\in\{0, 1\}^{\ell}$ and $\zeta, \eta\in\{0, 1\}^{(n-1)\ell}$ as follows:
\begin{itemize}
\item $\underline{G}(t)^{\sigma}$ and  $\underline{G}(t)^{\tau}$ are linearly independent,
\item ${\rm wt}(\sigma\oplus\tau)=\displaystyle\min_{u\neq v, u, v\in\{0, 1\}^{\ell}}\{{\rm wt}(u\oplus v) \mid \underline{G}(t)^{u}$ and $\underline{G}(t)^{v}$
  are linearly independent$\}$.
\end{itemize}
Let $x_{\beta}=\begin{pmatrix}
\underline{G}^{\beta\curvearrowleft_{t} \sigma}\\
\underline{G}^{\beta\curvearrowleft_{t} \tau}
\end{pmatrix}$, then there exist $\zeta, \eta$ such that:

\begin{itemize}
\item $x_{\zeta}$ and $x_{\eta}$ are linearly independent,
\item ${\rm wt}(\zeta\oplus \eta)=\displaystyle\min_{u\neq v, u, v\in\{0, 1\}^{(n-1)\ell}}\{{\rm wt}(u\oplus v) \mid x_{u}$  and $x_{v}$ are linearly independent$\}$.
\end{itemize}

By the definition of $\sigma, \tau, \zeta, \eta$, we directly have the following Lemma.
\begin{lemma}\label{lemma4.1.1}
If there exists $\alpha\in\{0, 1\}^{\ell}$ such that $0<{\rm wt}(\sigma\oplus \alpha)<{\rm wt}(\sigma\oplus \tau)$ and
$0<{\rm wt}(\alpha\oplus \tau)<{\rm wt}(\sigma\oplus \tau)$, then $\underline{G}(t)^{\alpha}$ is identically zero. Similarly, If there exists $\beta\in\{0, 1\}^{(n-1)\ell}$ such that $0<{\rm wt}(\eta\oplus \beta)<{\rm wt}(\eta\oplus \zeta)$ and
$0<{\rm wt}(\beta\oplus \zeta)<{\rm wt}(\eta\oplus \zeta)$, then $x_{\beta}$ is identically zero.
\end{lemma}


Let
$\zeta \oplus \eta = e_{q_{1}}\oplus e_{q_{2}}\oplus \cdots \oplus e_{q_{d'}}$,
and
$\sigma\oplus \tau=e_{p_{1}}\oplus e_{p_{2}}\oplus \cdots \oplus e_{p_{d}}$.

\begin{remark}
In this paper, we will repeatedly use the following method to construct MGI.
For any given  $\sigma, \tau, \zeta, \eta$, define 
$\{p_{1}, p_{2}, \ldots, p_{d}\}$ and
$\{q_{1}, {q_{2}}, \ldots, q_{d'}\}$ as above.
\begin{itemize}
\item Let the position vector be $(q_{1}q_{2}\cdots q_{d'})\curvearrowleft_{t}(p_{1}p_{2}\cdots p_{d})$ and the 
pattern  be $\zeta\curvearrowleft_{t} (\sigma \oplus e_{p_{1}})$,
where $p_{1}$ is the first non-zero position of the $t$-th block.
The other pattern is $\eta\curvearrowleft_{t} (\tau \oplus e_{p_{1}})$.
Then we will get a Matchgate Identity such that
the product
 $\underline{G}^{\zeta\curvearrowleft_{t} \sigma}\underline{G}^{\eta\curvearrowleft_{t}\tau}$
is term. Note that $(q_{1}q_{2}\cdots q_{d'})\curvearrowleft_{t}(p_{1}p_{2}\cdots p_{d})$ is the set of the non-zero positions of
$(\zeta\curvearrowleft_{t} \sigma)\oplus (\eta\curvearrowleft_{t} \tau)$.
\item Every bit position in $(q_{1}q_{2}\cdots q_{d'})\curvearrowleft_{t}(p_{1}p_{2}\cdots p_{d})$ corresponds to a product term. For the position $p_{i}$ in the $t$-th block, the product term is
$\underline{G}^{\zeta\curvearrowleft_{t} (\sigma\oplus e_{p_{1}}\oplus e_{p_{i}})}\underline{G}^{\eta\curvearrowleft_{t}(\tau\oplus e_{p_{1}}\oplus e_{p_{i}})}$.
    Outside the $t$-th block, the product term has the form
    $\underline{G}^{(\zeta\oplus e_{q_{j}})\curvearrowleft_{t} (\sigma\oplus e_{p_{1}})}\underline{G}^{(\eta\oplus e_{q_{j}})\curvearrowleft_{t}(\tau\oplus e_{p_{1}})}$.
\item Separate out the product terms corresponding to the positions in the $t$-th block from  the rest in MGI,  we get the equation
  \begin{equation*}
\displaystyle\sum_{i=1}^{d}(-1)^{i+ 1}\underline{G}^{\zeta\curvearrowleft_{t} (\sigma\oplus e_{p_{1}}\oplus e_{p_{i}})}\underline{G}^{\eta\curvearrowleft_{t} (\tau
\oplus e_{p_{1}}\oplus e_{p_{i}})}
=
\displaystyle\sum_{j=1}^{d'}(\pm\underline{G}^{(\zeta\oplus e_{q_{j}})\curvearrowleft_{t} (\sigma\oplus e_{p_{1}})}\underline{G}^{(\eta\oplus e_{q_{j}})\curvearrowleft_{t} (\tau\oplus e_{p_{1}})}),
\end{equation*}
where $\pm$ depends on $j$, and if $q_{j}$ occurs after
the $t$-th block then it also depends on the parity of $d$.
However this will not matter to us, since in this paper whenever
we use this method we will
show that these  terms where we write an indefinite $\pm$ sign all vanish.
\end{itemize}
\end{remark}

\begin{lemma}
\label{lemma4.2}
The Hamming distance $d$ between the minimizing
$\sigma$ and $\tau$ is  1.
\end{lemma}

\begin{proof}
For a contradiction assume $d\geq 2$.  Let $x_{\zeta}=\begin{pmatrix}
\underline{G}^{\zeta\curvearrowleft_{t} \sigma}\\
\underline{G}^{\zeta\curvearrowleft_{t} \tau}
\end{pmatrix}$,
$x_{\eta}=\begin{pmatrix}
\underline{G}^{\eta\curvearrowleft_{t} \sigma}\\
\underline{G}^{\eta\curvearrowleft_{t} \tau}
\end{pmatrix}$. These are linearly independent by definition.
 We will prove that $\underline{G}^{\zeta\curvearrowleft_{t} \sigma}\underline{G}^{\eta\curvearrowleft_{t} \tau}-
\underline{G}^{\zeta\curvearrowleft_{t} \tau}\underline{G}^{\eta\curvearrowleft_{t} \sigma}=0$ by MGI to get a contradiction.

We will apply MGI twice. The first time
we use the pattern $\zeta\curvearrowleft_{t} (\sigma\oplus e_{p_{1}})$
with the position vector $(q_{1}q_{2}\cdots q_{d'})\curvearrowleft_{t}(p_{1}p_{2}\cdots p_{d})$. Note that the other pattern obtained by XOR
is $\eta\curvearrowleft_{t}(\tau \oplus e_{p_{1}})$.
This gives
\begin{equation}\label{MGI4.1}
\displaystyle\sum_{i=1}^{d}(-1)^{i+1}\underline{G}^{\zeta\curvearrowleft_{t} (\sigma\oplus e_{p_{1}}\oplus e_{p_{i}})}\underline{G}^{\eta\curvearrowleft_{t} (\tau
\oplus e_{p_{1}}\oplus e_{p_{i}})}=
\displaystyle\sum_{j=1}^{d'}(\pm\underline{G}^{(\zeta\oplus e_{q_{j}})\curvearrowleft_{t} (\sigma\oplus e_{p_{1}})}\underline{G}^{(\eta\oplus e_{q_{j}})\curvearrowleft_{t} (\tau\oplus e_{p_{1}})}),
\end{equation}

The second time we use the same position vector $(q_{1}q_{2}\cdots q_{d'})\curvearrowleft_{t}(p_{1}p_{2}\cdots p_{d})$ with
the pattern $\zeta\curvearrowleft_{t} (\tau
\oplus e_{p_{1}})$.
Note that the other pattern obtained by XOR
is $\eta\curvearrowleft_{t} (\sigma\oplus e_{p_{1}})$. This gives

\begin{equation}\label{MGI4.2}
\displaystyle\sum_{i=1}^{d}(-1)^{i+1}\underline{G}^{\zeta\curvearrowleft_{t} (\tau
\oplus e_{p_{1}}\oplus e_{p_{i}})}\underline{G}^{\eta\curvearrowleft_{t} (\sigma\oplus e_{p_{1}}\oplus e_{p_{i}})}=
\displaystyle\sum_{j=1}^{d'}(\pm\underline{G}^{(\zeta\oplus e_{q_{j}})\curvearrowleft_{t} (\tau\oplus e_{p_{1}})}\underline{G}^{(\eta\oplus e_{q_{j}})\curvearrowleft_{t} (\sigma\oplus e_{p_{1}})}).
\end{equation}
 Note that (\ref{MGI4.1}) and (\ref{MGI4.2}) are symmetric: By switching $\sigma$ and $\tau$, we will go from (\ref{MGI4.1}) to (\ref{MGI4.2}).

If $d=2$, from (\ref{MGI4.1}), we have
\begin{center}
$\underline{G}^{\zeta\curvearrowleft_{t} \sigma}\underline{G}^{\eta\curvearrowleft_{t} \tau}-
\underline{G}^{\zeta\curvearrowleft_{t} \tau}\underline{G}^{\eta\curvearrowleft_{t} \sigma}=\displaystyle\sum_{j=1}^{d'}(\pm\underline{G}^{(\zeta\oplus e_{q_{j}})\curvearrowleft_{t} (\sigma\oplus e_{p_{1}})}\underline{G}^{(\eta\oplus e_{q_{j}})\curvearrowleft_{t} (\tau\oplus e_{p_{1}})}$).
\end{center}
From Lemma~\ref{lemma4.1.1}, by taking $\alpha = \tau\oplus e_{p_{1}}$,
we get
$\underline{G}^{(\eta\oplus e_{q_{j}})\curvearrowleft_{t} (\tau\oplus e_{p_{1}})}=0$ for $1\leq j\leq d'$.
Thus $\underline{G}^{\zeta\curvearrowleft_{t} \sigma}\underline{G}^{\eta
\curvearrowleft_{t} \tau}-\underline{G}^{\zeta\curvearrowleft_{t} \tau}\underline{G}^{\eta\curvearrowleft_{t} \sigma}=0$.
This contradicts that
$x_{\zeta}=\begin{pmatrix}
\underline{G}^{\zeta\curvearrowleft_{t} \sigma}\\
\underline{G}^{\zeta\curvearrowleft_{t} \tau}
\end{pmatrix}$ and
$x_{\eta}=\begin{pmatrix}
\underline{G}^{\eta\curvearrowleft_{t} \sigma}\\
\underline{G}^{\eta
\curvearrowleft \tau}
\end{pmatrix}$
are linearly independent, so $d\neq 2$.

If $d>2$, then $\underline{G}^{\zeta\curvearrowleft_{t} (\sigma\oplus e_{p_{1}}\oplus e_{p_{i}})}=0$
and $\underline{G}^{\zeta\curvearrowleft_{t} (\tau\oplus e_{p_{1}}\oplus e_{p_{i}})}=0$,
 for $i>1$ by Lemma \ref{lemma4.1.1}. From (\ref{MGI4.1})
and (\ref{MGI4.2}), we have
\begin{equation}\label{MGI4.3}
\underline{G}^{\zeta\curvearrowleft_{t} \sigma}\underline{G}^{\eta
\curvearrowleft_{t} \tau}=
\displaystyle\sum_{j=1}^{d'}(\pm\underline{G}^{(\zeta\oplus e_{q_{j}})\curvearrowleft_{t} (\sigma\oplus e_{p_{1}})}\underline{G}^{(\eta\oplus e_{q_{j}})\curvearrowleft_{t} (\tau\oplus e_{p_{1}})}),
\end{equation}

\begin{equation}\label{MGI4.4}
\underline{G}^{\zeta\curvearrowleft_{t} \tau}\underline{G}^{\eta\curvearrowleft_{t} \sigma}=
\displaystyle\sum_{j=1}^{d'}(\pm\underline{G}^{(\zeta\oplus e_{q_{j}})\curvearrowleft_{t} (\tau\oplus e_{p_{1}})}\underline{G}^{(\eta\oplus e_{q_{j}})\curvearrowleft_{t} (\sigma\oplus e_{p_{1}})}).
\end{equation}
In the right hand side of (\ref{MGI4.3}) and (\ref{MGI4.4}),
$\underline{G}^{(\eta\oplus e_{q_{j}})\curvearrowleft_{t} (\tau\oplus e_{p_{1}})}=0,$
$\underline{G}^{(\eta\oplus e_{q_{j}})\curvearrowleft_{t} (\sigma\oplus e_{p_{1}})}=0$ for $1\leq j\leq d'$ by Lemma~\ref{lemma4.1.1},
so $\underline{G}^{\zeta\curvearrowleft_{t} \sigma}\underline{G}^{\eta
\curvearrowleft_{t} \tau}=0$ and
$\underline{G}^{\zeta\curvearrowleft_{t} \tau}\underline{G}^{\eta\curvearrowleft_{t} \sigma}=0$.
This is also a contradiction to the linear independence of
$x_{\zeta}=\begin{pmatrix}
\underline{G}^{\zeta\curvearrowleft_{t} \sigma}\\
\underline{G}^{\zeta\curvearrowleft_{t} \tau}
\end{pmatrix}$ and
$x_{\eta}=\begin{pmatrix}
\underline{G}^{\eta\curvearrowleft_{t} \sigma}\\
\underline{G}^{\eta
\curvearrowleft \tau}
\end{pmatrix}$.
It follows that $d=1$.
\end{proof}

From Lemma \ref{lemma4.2},
we have $\sigma\oplus \tau=e_{p_{1}}$
and  $\zeta \oplus \eta = e_{q_{1}}\oplus e_{q_{2}}\oplus \cdots \oplus e_{q_{d'}}$.

\begin{lemma}
  \label{lemma4.4}
For the given $\sigma$ and $\tau$,
the Hamming distance $d'$ between the minimizing
$\zeta$ and $\eta$ is 1.
  \end{lemma}
\begin{proof}
For a contradiction assume $d'\geq 2$.  
 We apply MGI twice, with the same position vector $(q_{1}q_{2}\cdots q_{d'})\curvearrowleft_{t} p_{1}$ and
 the patterns $\zeta\curvearrowleft_{t}(\sigma\oplus e_{p_{1}})
= \zeta\curvearrowleft_{t} \tau$
and $\zeta\curvearrowleft_{t}(\tau\oplus e_{p_{1}})
= \zeta\curvearrowleft_{t} \sigma$ respectively.

\begin{equation}
\underline{G}^{\zeta\curvearrowleft_{t}\sigma}\underline{G}^{\eta
\curvearrowleft_{t}\tau}=
\displaystyle\sum_{j=1}^{d'}(\pm\underline{G}^{(\zeta\oplus e_{q_{j}})\curvearrowleft_{t}\tau}
\underline{G}^{(\eta\oplus e_{q_{j}})\curvearrowleft_{t}\sigma}),
\end{equation}
\begin{equation}
\underline{G}^{\zeta\curvearrowleft_{t}\tau}\underline{G}^{\eta
\curvearrowleft_{t}\sigma}=
\displaystyle\sum_{j=1}^{d'}(\pm\underline{G}^{(\zeta\oplus e_{q_{j}})\curvearrowleft_{t}\sigma}
\underline{G}^{(\eta\oplus e_{q_{j}})\curvearrowleft_{t}\tau}).
\end{equation}


By Lemma~\ref{lemma4.1.1},
we have $\underline{G}^{(\zeta\oplus e_{q_{j}})\curvearrowleft_{t}\tau}=0$
and $\underline{G}^{(\zeta\oplus e_{q_{j}})\curvearrowleft_{t}\sigma}=0$
 for $1\leq j\leq d'$.
 So
$\underline{G}^{\zeta\curvearrowleft_{t}\sigma}\underline{G}^{\eta
\curvearrowleft_{t}\tau}=0$ and
$\underline{G}^{\zeta\curvearrowleft_{t}\tau}\underline{G}^{\eta\curvearrowleft_{t}\sigma}=0$.
This contradicts that $x_{\zeta}=\begin{pmatrix}
\underline{G}^{\zeta\curvearrowleft_{t} \sigma}\\
\underline{G}^{\zeta\curvearrowleft_{t} \tau}
\end{pmatrix}$,
$x_{\eta
}=\begin{pmatrix}
\underline{G}^{\eta\curvearrowleft_{t} \sigma}\\
\underline{G}^{\eta\curvearrowleft_{t} \tau}
\end{pmatrix}$
are linearly independent.
Thus $d'=1$.
\end{proof}

From Lemma \ref{lemma4.2} and Lemma \ref{lemma4.4}
we have
$\sigma\oplus \tau=e_{p_{1}}$
and  $\zeta \oplus \eta = e_{q_{1}}$.

The basis $M$
is a $2^{\ell}\times 2$ matrix with rows indexed by 
$\alpha\in\{0, 1\}^{\ell}$. Let $M^{\alpha}$ denote
the row with index $\alpha$. 
Let $M^{\alpha_{1}, \alpha_{2}, \cdots, \alpha_{s}}$
denote the sub-matrix of $M$ whose rows are $\alpha_{1}, \alpha_{2}, \cdots, \alpha_{s}$.
Then we have the following corollary.
\begin{coro}\label{M-has-two-submatrix}
$M^{\sigma, \tau}$
is invertible.
\end{coro}
\begin{proof}
Since $\begin{pmatrix}
\underline{G}^{\zeta\curvearrowleft_{t} \sigma}&\underline{G}^{\eta\curvearrowleft_{t} \sigma}\\
\underline{G}^{\zeta\curvearrowleft_{t} \tau}&\underline{G}^{\eta\curvearrowleft_{t} \tau}
\end{pmatrix}$ is a sub-matrix of
$M^{\sigma, \tau}
G(t)(M^{\tt{T}})^{\otimes (n-1)}$,
$M^{\sigma, \tau}$ has rank at least 2.
Since $M^{\sigma, \tau}$ is a $2\times 2$ matrix, it
follows that $M^{\sigma, \tau}$ has rank exactly 2 and is invertible.
\end{proof}


Note that $(M^{\sigma, \tau})^{\otimes n}G$ is a column vector 
obtained 
by taking only entries of $\underline{G}$ with index values
either $\sigma$ or $\tau$.
It is a column vector  of dimension $2^{n}$ and we denote it by $\underline{G}^{*\leftarrow\sigma, \tau}$.

\begin{theorem}
\label{coro5.2}
$\underline{G}^{*\leftarrow\sigma, \tau}=(M^{\sigma, \tau})^{\otimes n}G$ is the standard signature of a generator matchgate of arity $n$ and
{\rm rank}($\underline{G}^{*\leftarrow\sigma, \tau}(t)$)=2 for some $t$.
\end{theorem}

\begin{proof}
Let $\Gamma$ be a matchgate realizing the standard signature $\underline{G}=M^{\otimes n}G$. Note that $\Gamma$ has $n\ell$ external nodes.
For every block with $\ell$ nodes, for  $1\leq i\leq \ell$,
 if the $i$-th bit of $\sigma$ is 1
then  we add an edge
 of weight 1 to the $i$-th external node, 
and the new node replaces it as an external node.
 If the $i$-th bit of $\sigma$ is 0
then we do nothing to it.
We get a new matchgate $\Gamma'$.
Next, we define  $\Gamma''$ from  $\Gamma'$:
 In each block of  $\ell$ external nodes of $\Gamma'$,
we pick only the ${p}_{1}$-th external node as an external node
of $\Gamma''$; all others are considered internal nodes of
$\Gamma''$.
Then we get a matchgate  $\Gamma''$
 realizing $\underline{G}^{*\leftarrow\sigma, \tau}=(M^{\sigma, \tau})^{\otimes n}G$.  Note that all of the bits of $\sigma, \tau$ are the same except the ${p}_{1}$-th bit.
Since $M^{\sigma, \tau}$ has rank 2, $\underline{G}^{*\leftarrow\sigma, \tau}(t)=M^{\sigma, \tau}G(t)(M^{\sigma, \tau})^{\tt{T}\otimes (n-1)}$ has rank 2
when $G(t)$ has rank 2, by Lemma \ref{lemma 4.1}.
\end{proof}

Note that $(M^{\sigma, \tau})^{\otimes(t-1)}\otimes M\otimes (M^{\sigma, \tau})^{\otimes(n-t)}G$ is a column vector of dimension $2^{n+\ell-1}$ and we denote it by $\underline{G}^{t^{c}\leftarrow\sigma, \tau}$.

\begin{lemma}
\label{tclemma}
$\underline{G}^{t^{c}\leftarrow\sigma, \tau}$
is the standard signature
of a generator matchgate of arity $n+ \ell-1$.
\end{lemma}
\begin{proof}
The proof of this lemma is similar to Theorem~\ref{coro5.2}.
Let $\Gamma$ be a matchgate realizing the standard signature $\underline{G}=M^{\otimes n}G$.  $\Gamma$ has $n\ell$ external nodes. 
We do nothing to the $t$-th block.
For the other blocks, we add an edge of weight 1 to the $i$-th external node if the $i$-th bit of $\sigma$ is 1 and do nothing to it if the $i$-th bit of $\sigma$ is 0 for $1\leq i\leq \ell$. Then we get a new matchgate $\Gamma'$.
Now take the external nodes of $\Gamma'$  in
the $t$-th block, and
the ${p}_{1}$-th external node in the other blocks, 
then we get a matchgate realizing $\underline{G}^{t^{c}\leftarrow\sigma, \tau}$.
\end{proof}

\begin{lemma}\label{gplemma2}
Let $\underline{G}=(\underline{G}^{i_{1}i_{2}\cdots i_{n}})$ be the standard signature of a generator matchgate $\Gamma$ of arity $n$.
If $\underline{G}$ has full rank and the $t$-th matrix form $\underline{G}(t)$ has rank 2, then  there exists a standard signature $\underline{R}$ realized by a recognizer matchgate of arity $n$  such that $\underline{G}(t)\underline{R}(t)=I_{2}$.
\end{lemma}
\begin{proof}
From Lemma \ref{lemma4.2} and  Lemma \ref{lemma4.4}, there is  a sub-matrix
 $A=\begin{pmatrix}
\underline{G}^{\alpha\curvearrowleft_{t}0}&\underline{G}^{\beta\curvearrowleft_{t}0}\\
\underline{G}^{\alpha\curvearrowleft_{t}1}&\underline{G}^{\beta\curvearrowleft_{t}1}
\end{pmatrix}$ of rank 2 in  $\underline{G}(t)$, where
$\alpha, \beta\in\{0, 1\}^{n-1}$ and ${\rm wt}(\alpha\oplus \beta)=1$.

By the Parity Condition, $\underline{G}^{\alpha\curvearrowleft_{t}0}=\underline{G}^{\beta\curvearrowleft_{t}1}=0$
or $\underline{G}^{\beta\curvearrowleft_{t}0}=\underline{G}^{\alpha\curvearrowleft_{t}1}=0$.
We prove the Lemma for
the case $\underline{G}^{\beta\curvearrowleft_{t}0}=\underline{G}^{\alpha\curvearrowleft_{t}1}=0$, i.e,
$A=\begin{pmatrix}
\underline{G}^{\alpha\curvearrowleft_{t}0}&0\\
0&\underline{G}^{\beta\curvearrowleft_{t}1}
\end{pmatrix}.$
 The other case is similar.

Let $\underline{R}$ be a vector of dimension $2^{n}$, where
\begin{equation*}\begin{pmatrix}
\underline{R}_{\alpha\curvearrowleft_{t}0}&0\\
0&\underline{R}_{\beta\curvearrowleft_{t}1}
\end{pmatrix}=A^{-1},
\end{equation*}
 and all other entries of $\underline{R}$ are zero.
It is obvious that $\underline{G}(t)\underline{R}(t)=I_{2}$ (Note that $R(t)$ is a $2^{n-2}\times 2$ matrix).
Furthermore, there are only two non-zero entries in $\underline{R}$ and the Hamming weight of XOR of their indices is 2, so it satisfies the Parity Condition and $\mathrm{MGI}$ by Lemma~\ref{lemma2.2}. Thus $\underline{R}$ is a standard signature.
\end{proof}

Let
$T=M(M^{\sigma, \tau})^{-1}$. Note that $T^{\sigma, \tau}=I_{2}$. Then
\begin{equation*}
\underline{G}=M^{\otimes n}G=T^{\otimes n}(M^{\sigma, \tau})^{\otimes n}G=T^{\otimes n}\underline{G}^{*\leftarrow\sigma, \tau}
\end{equation*}
and from that,
\begin{equation*}
\underline{G}^{t^{c}\leftarrow\sigma, \tau}=(T^{\sigma, \tau})^{\otimes(t-1)}\otimes T\otimes (T^{\sigma, \tau})^{\otimes(n-t)}
\cdot G^{*\leftarrow\sigma, \tau}.
\end{equation*}
The entries of $\underline{G}^{t^{c}\leftarrow\sigma, \tau}$ can be
 indexed by $i_{1}\ldots i_{t-1}i'_{1}\ldots i'_{\ell}i_{t+1}\ldots i_{n}$, where
$i_{j}, i'_{j'}\in\{0, 1\}$.
We denote the matrix form of $\underline{G}^{t^{c}\leftarrow\sigma, \tau}$ 
by $\underline{G}^{t^{c}\leftarrow\sigma, \tau}(t)$,
whose rows are indexed by $i'_{1}\ldots i'_{\ell}$ and columns
are indexed by $i_{1}\ldots i_{t-1}i_{t+1}\ldots i_{n}$.
 Then by Lemma~\ref{from vec to mat}, we have
\begin{equation}
\underline{G}^{t^{c}\leftarrow\sigma, \tau}(t)=
 T\underline{G}^{*\leftarrow\sigma, \tau}(t)((T^{\sigma, \tau})^{\tt{T}})^ {\otimes (n-1)}
 =T\underline{G}^{*\leftarrow\sigma, \tau}(t).
 \end{equation}

\begin{lemma}
\label{lemma4.7}
$T$ is the standard signature of a transducer matchgate
with $\ell$-output and 1-input.
\end{lemma}
\begin{proof}
By  Theorem \ref{coro5.2} and  Lemma \ref{gplemma2}, there exists a standard signature of a recognizer matchgate
 $\underline{R}$ such that $\underline{G}^{*\leftarrow\sigma, \tau}(t)\underline{R}(t)=I_{2}$.
 Let $\Gamma_{1}$ be the matchgate realizing $\underline{G}^{t^{c}\leftarrow\sigma, \tau}$ with output nodes
 $X_{1}, \ldots, X_{t-1}, Y_{1},  \ldots$, $Y_{\ell}, Z_{t+1}, \ldots$, $ Z_{n}$, and let $\Gamma_{2}$ be
 the matchgate realizing $\underline{R}$ with input nodes $W_{1}, W_{2},$
$ \ldots, W_{n}$. Then connect $X_{i}$ with $W_{i}$ for $1\leq i\leq t-1$ and
$Z_{i}$ with $W_{i}$ for $t+1\leq i\leq n$ by an edge with weight 1 respectively, we get a transducer matchgate $\Gamma$ with output nodes $Y_{1}, Y_{2}, \ldots, Y_{\ell}$ and input node $X_{t}$.
Then $T=\underline{G}^{t^{c}\leftarrow\sigma, \tau}(t)\underline{R}(t)$ is the standard signature of $\Gamma$.
Note that all the connections are made respecting the planarity condition.
\end{proof}

The proof of Lemma \ref{lemma4.7} is illustrated by Fig. 1.

The following theorem is the main theorem in \cite{string6};
the algebraic method in this section gives a simpler proof.
This method will be used in the next section to prove a similar
collapse theorem for domain size 3 and 4.
\begin{theorem}\label{collapse2}
Any holographic algorithm on a basis of size $\ell\geq 2$ and domain size 2 which employs at least one generator signature of full rank can be simulated
on a basis of size 1.
\end{theorem}
\begin{proof}
Let $\underline{R}_{i}M^{\otimes m_i}=R_{i}$ for $1\leq i\leq r$ and $\underline{G}_{j}=M^{\otimes n_j}G_{j}$ for $1\leq j\leq g$, where $R_{i}, G_{j}$
 are recognizer and generator signatures that a holographic algorithm employs and $\underline{R}_{i}, \underline{G}_{j}$ are standard signatures.
Without loss of generality, let $G_{1}$ be of full rank.
We define $\sigma$ and $\tau$ in terms of $G_{1}$ and  apply 
Corollary~\ref{M-has-two-submatrix}.
The basis $M$ has a full rank sub-matrix $M^{\sigma, \tau}$, 
where ${\rm wt}(\sigma\oplus \tau)=1$, and $T=M(M^{\sigma, \tau})^{-1}$ is the standard signature of a transducer matchgate by Lemma \ref{lemma4.7}.
 Let $\underline{R}_{i}'=\underline{R}_{i}T^{\otimes m_i}$,
 then
 \begin{equation*}
 \underline{R}'_{i}(M^{\sigma, \tau})^{\otimes m_i}=R_{i},~~~~
 \underline{G}_{j}^{*\leftarrow \sigma, \tau}=(M^{\sigma, \tau})^{\otimes n_j}G_{j},
 \end{equation*}
  for $1\leq i\leq r$ and $1\leq j\leq g$, 
where $\underline{R}_{i}'$
and
 $\underline{G}_{j}^{*\leftarrow \sigma, \tau}$  are standard signatures by Lemma \ref{product transformer} and Theorem \ref{coro5.2}. This implies that $R_{i}, G_{j}$ are simultaneously realized on the basis $M^{\sigma, \tau}$ of size 1.
\end{proof} 

\section{Collapse Theorems on Domain Size 3 and 4}
In this section, assume that $G$ is a generator signature of full rank on domain size $k\geq 3$, the basis $M$ is a $2^{\ell}\times k$ matrix of rank $k$, and $\underline{G}=M^{\otimes n}G$ is a standard signature of arity $n\ell$.
Thus
 there exists $t\in[n]$ such that rank$(G(t))=k$.  From Lemma \ref{from vec to mat}, we have $\underline{G}(t)=MG(t)(M^{\tt{T}})^{\otimes (n-1)}$. Then by Lemma \ref{lemma 4.1}, rank$(\underline{G}(t))=k$. Because $k \ge 3$,
 we can define $\sigma, \tau\in\{0,1\}^{\ell}, \zeta, \eta\in\{0,1\}^{(n-1)\ell}$ as follows:

\begin{itemize}
\item $\sigma$ and $\tau$ have the same parity,
\item $\underline{G}(t)^{\sigma}$ and  $\underline{G}(t)^{\tau}$ are linearly independent,
\item ${\rm wt}(\sigma\oplus\tau)=\displaystyle\min_{u\neq v, u, v\in\{0, 1\}^{\ell}}\{{\rm wt}(u\oplus v) \mid \underline{G}(t)^{u}$ and  $\underline{G}(t)^{v}$
have the same parity and  are linearly independent$\}$.
\end{itemize}
Let $x_{\beta}=\begin{pmatrix}
\underline{G}^{\beta\curvearrowleft_{t} \sigma}\\
\underline{G}^{\beta\curvearrowleft_{t} \tau}
\end{pmatrix}$, then there exist $\zeta, \eta$ such that:

\begin{itemize}
\item $x_{\zeta}$ and $x_{\eta}$ are linearly independent,
\item ${\rm wt}(\zeta\oplus \eta)=\displaystyle\min_{u\neq v, u, v\in\{0, 1\}^{(n-1)\ell}}\{{\rm wt}(u\oplus v) \mid x_{u}$ and $x_{v}$
are linearly independent$\}$.
\end{itemize}
Then we directly have the following lemma that is similar to Lemma \ref{lemma4.1.1}.
\begin{lemma}\label{lemma5.1}
If there exists $\alpha\in\{0, 1\}^{\ell}$ that has the same parity as $\sigma$ such that $0<{\rm wt}(\sigma\oplus \alpha)<{\rm wt}(\sigma\oplus \tau)$ and
$0<{\rm wt}(\alpha\oplus \tau)<{\rm wt}(\sigma\oplus \tau)$, then $\underline{G}(t)^{\alpha}$ is identically zero. 
Similarly, If there exists $\beta$ such that $0<{\rm wt}(\eta\oplus \beta)<{\rm wt}(\eta\oplus \zeta)$ and
$0<{\rm wt}(\beta\oplus \zeta)<{\rm wt}(\eta\oplus \zeta)$, then $x_{\beta}$ is identically zero.

\end{lemma}

Let
$\zeta \oplus \eta = e_{q_{1}}\oplus e_{q_{2}}\oplus \cdots \oplus e_{q_{d'}}$,
and
$\sigma\oplus \tau=e_{p_{1}}\oplus e_{p_{2}}\oplus \cdots \oplus e_{p_{d}}$.
Then we have the following lemma.


\begin{lemma}\label{lemma5.2}
The Hamming distance $d$ between the minimizing
$\sigma$ and $\tau$ is  2.
\end{lemma}
\begin{proof}
$d\neq 1$ since $\sigma$ and $\tau$ have the same parity.

For a contradiction assume
$d>2$. Let $x_{\zeta}=\begin{pmatrix}
\underline{G}^{\zeta\curvearrowleft_{t} \sigma}\\
\underline{G}^{\zeta\curvearrowleft_{t} \tau}
\end{pmatrix}$,
$x_{\eta}=\begin{pmatrix}
\underline{G}^{\eta\curvearrowleft_{t} \sigma}\\
\underline{G}^{\eta\curvearrowleft_{t} \tau}
\end{pmatrix}$.
These are linearly independent by definition.
 We will prove that $\underline{G}^{\zeta\curvearrowleft_{t} \sigma}\underline{G}^{\eta\curvearrowleft_{t} \tau}-
\underline{G}^{\zeta\curvearrowleft_{t} \tau}\underline{G}^{\eta\curvearrowleft_{t} \sigma}=0$ by MGI to get a contradiction.


We will apply MGI twice. The first time
we use the pattern $\zeta\curvearrowleft_{t} (\sigma\oplus e_{p_{1}})$
with the position vector $(q_{1}q_{2}\cdots q_{d'})\curvearrowleft_{t}(p_{1}p_{2}\cdots p_{d})$. Note that the other pattern obtained by XOR
is $\eta\curvearrowleft_{t}(\tau \oplus e_{p_{1}})$.
This gives
\begin{equation}\label{MGI5.1}
\displaystyle\sum_{i=1}^{d}(-1)^{i+1}\underline{G}^{\zeta\curvearrowleft_{t} (\sigma\oplus e_{p_{1}}\oplus e_{p_{i}})}\underline{G}^{\eta\curvearrowleft_{t} (\tau
\oplus e_{p_{1}}\oplus e_{p_{i}})}=
\displaystyle\sum_{j=1}^{d'}(\pm\underline{G}^{(\zeta\oplus e_{q_{j}})\curvearrowleft_{t} (\sigma\oplus e_{p_{1}})}\underline{G}^{(\eta\oplus e_{q_{j}})\curvearrowleft_{t} (\tau\oplus e_{p_{1}})}).
\end{equation}

%

The second time we use the same position vector $(q_{1}q_{2}\cdots q_{d'})\curvearrowleft_{t}(p_{1}p_{2}\cdots p_{d})$ with
the pattern $\zeta\curvearrowleft_{t} (\tau
\oplus e_{p_{1}})$.
Note that the other pattern obtained by XOR
is $\eta\curvearrowleft_{t} (\sigma\oplus e_{p_{1}})$. This gives


\begin{equation}\label{MGI5.2}
\displaystyle\sum_{i=1}^{d}(-1)^{i+1}\underline{G}^{\zeta\curvearrowleft_{t} (\tau
\oplus e_{p_{1}}\oplus e_{p_{i}})}\underline{G}^{\eta\curvearrowleft_{t} (\sigma\oplus e_{p_{1}}\oplus e_{p_{i}})}=
\displaystyle\sum_{j=1}^{d'}(\pm\underline{G}^{(\zeta\oplus e_{q_{j}})\curvearrowleft_{t} (\tau\oplus e_{p_{1}})}\underline{G}^{(\eta\oplus e_{q_{j}})\curvearrowleft_{t} (\sigma\oplus e_{p_{1}})}),
\end{equation}

Since $d>2$, $\underline{G}^{\zeta\curvearrowleft_{t} (\sigma\oplus e_{p_{1}}\oplus e_{p_{i}})}=0$,
 $\underline{G}^{\zeta\curvearrowleft_{t} (\sigma\oplus e_{p_{1}}\oplus e_{p_{i}})}=0$ for $i>1$ from Lemma \ref{lemma5.1}. Then from (\ref{MGI5.1}) and (\ref{MGI5.2}) we have

\begin{equation}\label{MGI5.3}
\underline{G}^{\zeta\curvearrowleft_{t} \sigma}\underline{G}^{\eta\curvearrowleft_{t} \tau}=
\displaystyle\sum_{j=1}^{d'}(\pm\underline{G}^{(\zeta\oplus e_{q_{j}})\curvearrowleft_{t} (\sigma\oplus e_{p_{1}})}\underline{G}^{(\eta\oplus e_{q_{j}})\curvearrowleft_{t} (\tau\oplus e_{p_{1}})}),
\end{equation}

\begin{equation}\label{MGI5.4}
\underline{G}^{\zeta\curvearrowleft_{t} \tau}\underline{G}^{\eta\curvearrowleft_{t} \sigma}=
\displaystyle\sum_{j=1}^{d'}(\pm\underline{G}^{(\zeta\oplus e_{q_{j}})\curvearrowleft_{t} (\tau\oplus e_{p_{1}})}\underline{G}^{(\eta\oplus e_{q_{j}})\curvearrowleft_{t} (\sigma\oplus e_{p_{1}})}).
\end{equation}
Note that (\ref{MGI5.3}) and (\ref{MGI5.4}) are symmetric for $\sigma$ and $\tau$. By switching $\sigma$ and $\tau$, we will go from (\ref{MGI5.3}) to (\ref{MGI5.4}).
$x_{\zeta}=\begin{pmatrix}
\underline{G}^{\zeta\curvearrowleft_{t} \sigma}\\
\underline{G}^{\zeta\curvearrowleft_{t} \tau}
\end{pmatrix}\neq
\begin{pmatrix}
0\\0
\end{pmatrix}$
since
$x_{\zeta}=\begin{pmatrix}
\underline{G}^{\zeta\curvearrowleft_{t} \sigma}\\
\underline{G}^{\zeta\curvearrowleft_{t} \tau}
\end{pmatrix}$ and
$x_{\eta}=\begin{pmatrix}
\underline{G}^{\eta\curvearrowleft_{t} \sigma}\\
\underline{G}^{\eta\curvearrowleft_{t} \tau}
\end{pmatrix}$
are linearly independent. Assume that $\underline{G}^{\zeta\curvearrowleft_{t} \sigma}\neq 0$.
Let the pattern be $\zeta\curvearrowleft_{t}(\sigma\oplus e_{p_{2}})$ and the position vector be
$(q_{1}\cdots \hat{q}_{j}\cdots q_{d'})\curvearrowleft_{t}(p_{2}\cdots p_{d})$ for $1\leq j\leq d'$, where $q_{1}\cdots \hat{q}_{j}\cdots q_{d'}$ means deleting
$q_{j}$ from $q_{1}\cdots q_{j}\cdots q_{d'}$.
Note that the other pattern obtained by XOR is
$(\eta \oplus e_{q_{j}}) \curvearrowleft_{t} (\tau\oplus e_{p_{1}} \oplus e_{p_{2}})$.
%
Then we have
\begin{eqnarray*}\label{MGI5.6}
\begin{split}
&\underline{G}^{\zeta\curvearrowleft_{t} \sigma} \underline{G}^{(\eta\oplus e_{q_{j}})\curvearrowleft_{t} (\tau\oplus e_{p_{1}})}+
\displaystyle\sum_{i=3}^{d}(-1)^{i}\underline{G}^{\zeta\curvearrowleft_{t} (\sigma\oplus e_{p_{2}}\oplus e_{p_{i}})} \underline{G}^{(\eta
\oplus e_{q_{j}})
\curvearrowleft_{t} (\tau\oplus e_{p_{1}}\oplus e_{p_{2}}\oplus e_{p_{i}})}
\\
&=\displaystyle\sum_{1\leq u\leq d',u\neq j}(\pm\underline{G}^{(\zeta\oplus e_{q_{u}})
\curvearrowleft_{t} (\sigma\oplus e_{p_{2}})} \underline{G}^{(\eta\oplus e_{q_{j}}\oplus e_{q_{u}})\curvearrowleft_{t} (\tau\oplus e_{p_{1}}\oplus e_{p_{2}})}).
\end{split}
\end{eqnarray*}
$\underline{G}^{\zeta\curvearrowleft_{t} (\sigma\oplus e_{p_{2}}\oplus e_{p_{i}})}=0$ for $i>2$ and
$\underline{G}^{(\eta\oplus e_{q_{j}}\oplus e_{q_{u}})\curvearrowleft_{t} (\tau\oplus e_{p_{1}}\oplus e_{p_{2}})}=0$ from Lemma \ref{lemma5.1},
and by assumption $\underline{G}^{\zeta\curvearrowleft_{t} \sigma}\neq 0$,
 so
$\underline{G}^{(\eta\oplus e_{q_{j}})\curvearrowleft_{t} (\tau\oplus e_{p_{1}})}=0$ for $1\leq j\leq d'$.
Thus
\begin{equation}\label{eq5.1}
\underline{G}^{\zeta\curvearrowleft_{t} \sigma}\underline{G}^{\eta\curvearrowleft_{t} \tau}=0
\end{equation}
by (\ref{MGI5.3}).

Furthermore, let the pattern be $(\zeta\oplus e_{q_{j}})\curvearrowleft_{t}\sigma$
  and the position vector be $q_{j}\curvearrowleft_{t} p_{2}p_{3}\cdots p_{d}$
 for $1\leq j\leq d'$.
The other pattern obtained by XOR
is $\zeta\curvearrowleft_{t}(\tau\oplus e_{p_{1}})$). This gives a matchgate identity
%
\begin{eqnarray*}\label{MGI5.5}
\underline{G}^{\zeta\curvearrowleft_{t} \sigma} \underline{G}^{(\zeta\oplus e_{q_{j}})\curvearrowleft_{t} (\tau\oplus e_{p_{1}})}=\pm\displaystyle\sum_{i=2}^{d}(-1)^{i}\underline{G}^{(\zeta\oplus q_{j})\curvearrowleft_{t} (\sigma\oplus e_{p_{i}})} \underline{G}^{\zeta\curvearrowleft_{t} (\tau\oplus e_{p_{1}}\oplus e_{p_{i}})}.
\end{eqnarray*}
Since $\underline{G}^{\zeta\curvearrowleft_{t} (\tau\oplus e_{p_{1}}\oplus e_{p_{i}})}=0$ for $2\leq i\leq d$ from Lemma \ref{lemma5.1},
and by assumption $\underline{G}^{\zeta\curvearrowleft_{t} \sigma}\neq 0$,
we have
 $\underline{G}^{(\zeta\oplus e_{q_{j}})\curvearrowleft_{t} (\tau\oplus e_{p_{1}})}=0$ for $1\leq j\leq d'$.
Thus \begin{equation}\label{eq5.2}
\underline{G}^{\zeta\curvearrowleft_{t} \tau}\underline{G}^{\eta\curvearrowleft_{t} \sigma}=0
\end{equation} by (\ref{MGI5.4}).

This implies that
$\underline{G}^{\zeta\curvearrowleft_{t} \sigma}\underline{G}^{\eta\curvearrowleft_{t} \tau}-
\underline{G}^{\eta\curvearrowleft_{t} \sigma}\underline{G}^{\zeta\curvearrowleft_{t} \tau}=0$ by (\ref{eq5.1}) and (\ref{eq5.2}).
This contradicts that $x_{\zeta}=\begin{pmatrix}
\underline{G}^{\zeta\curvearrowleft_{t} \sigma}\\
\underline{G}^{\zeta\curvearrowleft_{t} \tau}
\end{pmatrix}$ and
$x_{\eta
}=\begin{pmatrix}
\underline{G}^{\eta\curvearrowleft_{t} \sigma}\\
\underline{G}^{\eta\curvearrowleft_{t} \tau}
\end{pmatrix}$
are linearly independent.
Hence $d=2$ under the hypothesis
$\underline{G}^{\zeta\curvearrowleft_{t} \sigma}\neq 0$.
By symmetry, switching $\sigma$ and $\tau$, we can also prove $d=2$ under the hypothesis $\underline{G}^{\zeta\curvearrowleft_{t} \tau}\neq 0$.
\end{proof}

From Lemma \ref{lemma5.2},
we have $\sigma\oplus \tau=e_{p_{1}} \oplus e_{p_{2}}$
and  $\zeta \oplus \eta = e_{q_{1}}\oplus e_{q_{2}}\oplus \cdots \oplus e_{q_{d'}}$.
\begin{lemma}\label{lemma5.3}
For the given $\sigma$ and $\tau$,
the Hamming distance $d'$ between the minimizing
$\zeta$ and $\eta$ is 2.
\end{lemma}
\begin{proof}
By the Parity Condition,
$d'\neq 1$.

For a contradiction assume $d'>2$.
Let the pattern be $\zeta\curvearrowleft_{t}(\sigma\oplus e_{p_{1}})$ and
the position vector be $(q_{1}q_{2}\cdots q_{d'})\curvearrowleft_{t}(p_{1}p_{2})$.
The other pattern obtained by XOR
is $\eta \curvearrowleft_{t} (\sigma\oplus e_{p_{2}})
= \eta \curvearrowleft_{t} (\tau \oplus e_{p_{1}})$.
\begin{equation}\label{MGI5.7}
\underline{G}^{\zeta\curvearrowleft_{t}\sigma}\underline{G}^{\eta\curvearrowleft_{t}\tau
}-
\underline{G}^{\zeta\curvearrowleft_{t}\tau}\underline{G}^{\eta\curvearrowleft_{t}\sigma}=
\displaystyle\sum_{j=1}^{d'}(\pm\underline{G}^{(\zeta\oplus e_{q_{j}})\curvearrowleft_{t}(\sigma\oplus e_{p_{1}})}
\underline{G}^{(\eta
\oplus e_{q_{j}})\curvearrowleft_{t}(\tau
\oplus e_{p_{1}})}).
\end{equation}

Since
$x_{\zeta}=\begin{pmatrix}
\underline{G}^{\zeta\curvearrowleft_{t} \sigma}\\
\underline{G}^{\zeta\curvearrowleft_{t} \tau}
\end{pmatrix}$ and
$x_{\eta}=\begin{pmatrix}
\underline{G}^{\eta\curvearrowleft_{t} \sigma}\\
\underline{G}^{\eta\curvearrowleft_{t} \tau}
\end{pmatrix}$
are linearly independent, we have
$\underline{G}^{\zeta\curvearrowleft_{t} \sigma}\neq 0$ or
$\underline{G}^{\eta\curvearrowleft_{t} \sigma}\neq 0$.

Assume  $\underline{G}^{\zeta\curvearrowleft_{t} \sigma}\neq 0$. Let the pattern be $\zeta\curvearrowleft_{t} (\sigma\oplus e_{p_{2}})$
and the position vector be $(q_{1}\cdots \hat{q}_{j}\cdots q_{d'})\curvearrowleft_{t}p_{2}$ for $1\leq j\leq d'$.
The other pattern obtained by XOR
is  $(\eta\oplus e_{q_{j}})\curvearrowleft_{t} \sigma$.
Then from MGI we have

\begin{equation}
\label{MGI5.8}
\underline{G}^{\zeta\curvearrowleft_{t}\sigma}\underline{G}^{(\eta
\oplus e_{q_{j}})\curvearrowleft_{t}(\sigma
\oplus e_{p_{2}})}
=\displaystyle\sum_{1\leq u\leq d',u\neq j}(\pm\underline{G}^{(\zeta\oplus e_{q_{u}})\curvearrowleft_{t}(\sigma\oplus e_{p_{2}})}
\underline{G}^{(\eta\oplus e_{q_{u}}\oplus e_{q_{j}})\curvearrowleft_{t}\sigma}).
\end{equation}
$\underline{G}^{(\eta\oplus e_{q_{u}}\oplus e_{q_{j}})\curvearrowleft_{t}\sigma}=0$ for $u\neq j$ by Lemma \ref{lemma5.1}, so $\underline{G}^{(\eta\oplus e_{q_{j}})\curvearrowleft_{t}(\tau\oplus e_{p_{1}})}=0$ for $1\leq j\leq d'$ from (\ref{MGI5.8}) and
$\tau\oplus e_{p_{1}}=\sigma\oplus e_{p_{2}}$. Then
$\underline{G}^{\zeta\curvearrowleft_{t}\sigma}\underline{G}^{\eta\curvearrowleft_{t}\tau}-
\underline{G}^{\zeta\curvearrowleft_{t}\tau}\underline{G}^{\eta\curvearrowleft_{t}\sigma}=0$ from (\ref{MGI5.7}).
This contradicts that $x_{\zeta}$ and $x_{\eta}$ are linearly independent.
Hence $d'=2$ under the hypothesis
$\underline{G}^{\zeta\curvearrowleft_{t} \sigma}\neq 0$.
By symmetry, switching $\zeta$ and $\eta$, we can also prove $d'=2$ under the hypothesis $\underline{G}^{\eta\curvearrowleft_{t} \sigma}\neq 0$.
%
%
%
\end{proof}

From Lemma \ref{lemma5.2} and Lemma \ref{lemma5.3},
we have
$\sigma\oplus \tau=e_{p_{1}} \oplus e_{p_{2}}$
and  $\zeta \oplus \eta = e_{q_{1}} \oplus e_{q_{2}}$.

\begin{theorem}\label{sub-basis}
Suppose
$G$ is a generator signature of full rank on domain size $k\geq 3$ with
 the $t$-th matrix form $G(t)$ having rank $k$.
Furthermore suppose there is a
 basis $M$ which is a $2^{\ell}\times k$ matrix of rank $k$, and $\underline{G}=M^{\otimes n}G$ is a standard signature of arity $n\ell$. Then
there is a $4\times 4$ sub-matrix of rank 4 in $\underline{G}(t)$.
\end{theorem}
\begin{proof}
Following the notations of Lemma \ref{lemma5.2} and Lemma \ref{lemma5.3}, the sub-matrix
$\begin{pmatrix}
\underline{G}^{\zeta\curvearrowleft_{t} \sigma}&\underline{G}^{\eta\curvearrowleft_{t} \sigma}\\
\underline{G}^{\zeta\curvearrowleft_{t} \tau}&\underline{G}^{\eta\curvearrowleft_{t} \tau}
\end{pmatrix}$  of $\underline{G}(t)$ has rank 2, where ${\rm wt}(\sigma\oplus \tau)=2$ and ${\rm wt}(\eta\oplus \zeta)=2$.
Let the pattern be $\zeta\curvearrowleft_{t} (\sigma\oplus e_{p_{1}})$
 and the position vector be $(q_{1}q_{2})\curvearrowleft_{t}(p_{1}p_{2})$.
The other pattern obtained by XOR
is $\eta \curvearrowleft_{t} (\sigma\oplus e_{p_{2}})
= \eta \curvearrowleft_{t} (\tau \oplus e_{p_{1}})$.
Then from
MGI we have
\begin{eqnarray*}
\underline{G}^{\zeta\curvearrowleft_{t} \sigma}\underline{G}^{\eta\curvearrowleft_{t} \tau}-
\underline{G}^{\zeta\curvearrowleft_{t} \tau}\underline{G}^{\eta\curvearrowleft_{t} \sigma}=
\pm(\underline{G}^{(\zeta\oplus e_{q_{1}})\curvearrowleft_{t} (\sigma\oplus e_{p_{1}})}\underline{G}^{(\eta\oplus e_{q_{1}})\curvearrowleft_{t} (\tau\oplus e_{p_{1}})}
-\underline{G}^{(\zeta\oplus e_{q_{2}})\curvearrowleft_{t} (\sigma\oplus e_{p_{1}})}\underline{G}^{(\eta\oplus e_{q_{2}})\curvearrowleft_{t} (\tau\oplus e_{p_{1}})}).
\end{eqnarray*}
It is $``+"$ if $q_{1}<p_{1}<p_{2}<q_{2}$. Otherwise, it is $``-"$,
when $q_{1}<q_{2}<p_{1}<p_{2}$ or $p_{1}<p_{2}<q_{1}<q_{2}$.
Thus
\begin{equation}\label{det}
\underline{G}^{(\zeta\oplus e_{q_{1}})\curvearrowleft_{t} (\sigma\oplus e_{p_{1}})}\underline{G}^{(\eta\oplus e_{q_{1}})\curvearrowleft_{t} (\tau\oplus e_{p_{1}})}
-\underline{G}^{(\zeta\oplus e_{q_{2}})\curvearrowleft_{t} (\sigma\oplus e_{p_{1}})}\underline{G}^{(\eta\oplus e_{q_{2}})\curvearrowleft_{t} (\tau\oplus e_{p_{1}})}\neq 0.
\end{equation}


By renaming $\sigma$ and $\tau$,
the $p_{1}, p_{2}$-th bits of $\sigma$ can be 00 or 01
(and the $p_{1}, p_{2}$-th bits of $\tau$ are 11 or 10 correspondingly).
Similarly by renaming $\zeta$ and $\eta$,
the $q_{1}, q_{2}$-th bits of $\zeta$ can be 00 or 01 (and the
$q_{1}, q_{2}$-th bits of $\eta$ are 11 or 10 correspondingly).
So there are four cases for the $q_{1}, q_{2}$-th, $p_{1}, p_{2}$-th bits of $\zeta\curvearrowleft_{t} \sigma$.
We may temporarily make the assumption that $q_{1}, q_{2}$-th, $p_{1}, p_{2}$-th bits of $\zeta\curvearrowleft_{t} \sigma$ are 00, 00 respectively.
Then $\underline{G}(t)$ has a full rank sub-matrix of the following form by (\ref{det}):
\begin{equation}\label{5.4.1}
\begin{pmatrix}
\underline{G}^{\zeta\curvearrowleft_{t} \sigma}&0&0&
\underline{G}^{\eta\curvearrowleft_{t} \sigma}\\
0&\underline{G}^{(\zeta\oplus e_{q_{2}})\curvearrowleft_{t} (\sigma\oplus e_{p_{2}})}&\underline{G}^{(\zeta\oplus e_{q_{1}})\curvearrowleft_{t} (\sigma\oplus e_{p_{2}})}&0\\
0&\underline{G}^{(\zeta\oplus e_{q_{2}})\curvearrowleft_{t} (\sigma\oplus e_{p_{1}})}&\underline{G}^{(\zeta\oplus e_{q_{1}})\curvearrowleft_{t} (\sigma\oplus e_{p_{1}})}&0\\
\underline{G}^{\zeta\curvearrowleft_{t} \tau}&0&0&
\underline{G}^{\eta\curvearrowleft_{t} \tau}
\end{pmatrix}.
\end{equation}

If the $q_{1}, q_{2}$-th, $p_{1}, p_{2}$-th bits of $\zeta\curvearrowleft_{t}\sigma$ are 00,01, then the full rank sub-matrix has the following form by (\ref{det}):
\begin{equation}\label{5.4.2}
\begin{pmatrix}
0&\underline{G}^{(\zeta\oplus e_{q_{2}})\curvearrowleft_{t} (\sigma\oplus e_{p_{2}})}&\underline{G}^{(\zeta\oplus e_{q_{1}})\curvearrowleft_{t} (\sigma\oplus e_{p_{2}})}&0\\
\underline{G}^{\zeta\curvearrowleft_{t} \sigma}&0&0&\underline{G}^{\eta\curvearrowleft_{t} \sigma}\\
\underline{G}^{\zeta\curvearrowleft_{t} \tau}&0&0&\underline{G}^{\eta\curvearrowleft_{t} \tau}\\
0&\underline{G}^{(\zeta\oplus e_{q_{2}})\curvearrowleft_{t} (\sigma\oplus e_{p_{1}})}&\underline{G}^{(\zeta\oplus e_{q_{1}})\curvearrowleft_{t} (\sigma\oplus e_{p_{1}})}&0
\end{pmatrix}.
\end{equation}

If the $q_{1}, q_{2}$-th, $p_{1}, p_{2}$-th bits of $\zeta\curvearrowleft_{t}\sigma$ are 01,00, then we get the full rank sub-matrix by switching
$\zeta$ with $\zeta\oplus e_{q_{2}}$ and $\sigma$ with $\sigma\oplus e_{p_{2}}$ in (\ref{5.4.2}).

If the $q_{1}, q_{2}$-th, $p_{1}, p_{2}$-th bits of $\zeta\curvearrowleft_{t}\sigma$ are 01,01, then we get the full rank sub-matrix by switching
$\zeta$ with $\zeta\oplus e_{q_{2}}$ and $\sigma$ with $\sigma\oplus e_{p_{2}}$ in (\ref{5.4.1}).
\end{proof}

\begin{coro}\label{rank}
Let $\underline{G}$ be a standard signature and $\underline{G}(t)$ be the $t$-th matrix form of $\underline{G}$. If there are two rows with the same parity that are linearly independent in $\underline{G}(t)$, then rank($\underline{G}(t))\geq 4$.
\end{coro}

\begin{remark}
If $\underline{G}(t)$ has rank $k\geq 3$,
then there are some two rows, of the same parity, that are linearly
independent.
\end{remark}

\begin{coro}\label{rank3}
If there exists $t\in[k]$ such that rank$(G(t))=3$ and $M$ is a $2^{\ell}\times k$ matrix of rank $k$, then $G$ \emph{cannot} be realized on $M$.
\end{coro}
\begin{proof}
If $\underline{G}=M^{\otimes n}G$, then $\underline{G}(t)=MG(t)(M^{\tt{T}})^{\otimes (n-1)}$ and
rank$(\underline{G}(t))=3$ by Lemma \ref{lemma 4.1}.
So there are two rows that have the same parity and are linearly independent in $\underline{G}(t)$.
Then by Corollary \ref{rank}, rank$(\underline{G}(t))\geq 4$. This is a contradiction, so $G$ cannot be realized on any full rank basis.
\end{proof}

\subsection{A Collapse Theorem on Domain Size 3}
From Corollary \ref{rank3}, we directly have the following Theorem:
\begin{theorem}
Let $G$ be a full rank generator signature on domain size 3, then it 
\emph{cannot} be realized on a  basis of rank 3.
\end{theorem}

\begin{remark}
So for generator signatures on domain size 3, the only way to be realizable by a full ranked
matchgate
is via a basis of rank 2.
\end{remark}

In paper $\cite{string8}$, the following result  about holographic algorithms on bases of rank 2 is given:
\begin{theorem}
 Let $R_{1}, R_{2}, \cdots, R_{r}$ and $G_{1}, G_{2}, \cdots, G_{g}$ be the recognizer and generator signatures on domain size $k\geq 3$ that a holographic algorithm employs.
    If $R_{1}, R_{2}, \cdots, R_{r}$ and $G_{1}, G_{2}, \cdots, G_{g}$ can be realized on a  $2^{\ell}\times k$ basis $M$ of rank 2, where $\ell\geq 2$, then we can efficiently find signatures $R'_{1}, R'_{2}, \cdots, R'_{r}$ and $G'_{1}, G'_{2}, \cdots, G'_{g}$ on domain size 2 such that
the {\it contraction} of $\bigotimes_{i=1}^{g}G'_{i}$ and $\bigotimes_{j=1}^{r}R'_{j}$ is equal to the {\it contraction} of $\bigotimes_{i=1}^{g}G_{i}$ and $\bigotimes_{j=1}^{r}R_{j}$. And if there is at least one full rank generator signature $G'_{i}$, then
$R_{1}, R_{2}, \cdots, R_{r}$ and $G_{1}, G_{2}, \cdots, G_{g}$ can be simultaneously realized on a $2\times k$ basis (of size 1).
\end{theorem}


\begin{remark}
Since $G'_{1}, G'_{2}, \cdots, G'_{g}$ are signatures on domain size 2, 
either they are all degenerate or there is some $G'_{i}$ that is of full rank by Definition \ref{def3.1} and Definition \ref{def3.2}.
If $G'_{1}, G'_{2}, \cdots, G'_{g}$ are all degenerate, then
the contraction of $\bigotimes_{i=1}^{g}G'_{i}$ and $\bigotimes_{j=1}^{r}R'_{j}$
can be computed trivially in polynomial time.
And since the {\it contraction} of $\bigotimes_{i=1}^{g}G'_{i}$ and $\bigotimes_{j=1}^{r}R'_{j}$ is same as the {\it contraction} of $\bigotimes_{i=1}^{g}G_{i}$ and $\bigotimes_{j=1}^{r}R_{j}$, the holographic algorithm can
be computed trivially in polynomial time.
Otherwise we have the following collapse theorem for holographic algorithms on domain size 3.
\end{remark}

\begin{theorem}
On domain size 3,
for any 
holographic algorithm using 
a set of signatures $R_{1}, R_{2}, \cdots, R_{r}$ and $G_{1}, G_{2}, \cdots, G_{g}$, and a $2^{\ell}\times 3$ basis
$M$ of size $\ell\geq 2$
on which these signatures are simultaneously realized by matchgates,
we can efficiently replace the signatures with an equivalent set
having the same Holant value.
And for the new set of signatures, we can efficiently decide whether
it is a degenerate set, i.e., all generators are degenerate, in which case
we can compute the Holant in polynomial time trivially,
or else,
we can efficiently find another $2\times 3$ basis
$M'$ (of size $1$),
on which $R_{1}, R_{2}, \cdots, R_{r}$ and $G_{1}, G_{2}, \cdots, G_{g}$ can be simultaneously realized.
Thus, any non-trivial holographic algorithm using matchgates
on domain size 3 can be accomplished by a basis of size 1.
\end{theorem}

\subsection{A Group Property of Standard Signatures}
\begin{lemma}\label{even arity 4}
Assume that
\begin{equation*}
A=\begin{pmatrix}
G^{0000}&0&0&G^{0011}\\
0&G^{0101}&G^{0110}&0\\
0&G^{1001}&G^{1010}&0\\
G^{1100}&0&0&G^{1111}\\
\end{pmatrix}
\end{equation*}
has rank 4 and
\begin{equation}\label{pm-mgi}
G^{0000}G^{1111}-G^{1100}G^{0011}=\pm(G^{1010}G^{0101}-G^{1001}G^{0110}),
\end{equation}
then $G^{-1}$ is of the following form
\begin{equation*}
G^{-1}=\begin{pmatrix}
g^{0000}&0&0&g^{0011}\\
0&g^{0101}&g^{0110}&0\\
0&g^{1001}&g^{1010}&0\\
g^{1100}&0&0&g^{1111}
\end{pmatrix}
\end{equation*}
and
\begin{equation*}
g^{0000}g^{1111}-g^{1100}g^{0011}=\pm(g^{1010}g^{0101}-g^{1001}g^{0110}),
\end{equation*}
with the same sign $\pm$ as in (\ref{pm-mgi}).
\end{lemma}
\begin{proof}
Note that
$\begin{pmatrix}
G^{0000}&G^{0011}\\
G^{1100}&G^{1111}
\end{pmatrix}$ and $\begin{pmatrix}
G^{0101}&G^{0110}\\
G^{1001}&G^{1010}
\end{pmatrix}$
have rank 2.
Let
\begin{equation*}
\begin{pmatrix}
g^{0000}&g^{0011}\\
g^{1100}&g^{1111}
\end{pmatrix}
=
\begin{pmatrix}
G^{0000}&G^{0011}\\
G^{1100}&G^{1111}
\end{pmatrix}^{-1},
\begin{pmatrix}
g^{0101}&g^{0110}\\
g^{1001}&g^{1010}
\end{pmatrix}
=
\begin{pmatrix}
G^{0101}&G^{0110}\\
G^{1001}&G^{1010}
\end{pmatrix}^{-1},
\end{equation*}
then
\begin{equation*}
G^{-1}=\begin{pmatrix}
g^{0000}&0&0&g^{0011}\\
0&g^{0101}&g^{0110}&0\\
0&g^{1001}&g^{1010}&0\\
g^{1100}&0&0&g^{1111}
\end{pmatrix}
\end{equation*}
and
\begin{equation*}
\begin{split}
&g^{0000}g^{1111}-g^{1100}g^{0011}=(G^{0000}G^{1111}-G^{1100}G^{0011})^{-1},\\
&g^{1010}g^{0101}-g^{1001}g^{0110}=(G^{1010}G^{0101}-G^{1001}G^{0110})^{-1}.
\end{split}
\end{equation*}
Thus
\begin{equation*}
g^{0000}g^{1111}-g^{1100}g^{0011}=\pm(g^{1010}g^{0101}-g^{1001}g^{0110}).
\end{equation*}
\end{proof}

Similarly, we have the following Lemma.

\begin{lemma}\label{odd arity 4}
Assume that
\begin{equation*}
A=\begin{pmatrix}
0&G^{0001}&G^{0010}&0\\
G^{0100}&0&0&G^{0111}\\
G^{1000}&0&0&G^{1011}\\
0&G^{1101}&G^{1110}&0
\end{pmatrix}
\end{equation*}
has rank 4 and
\begin{equation}\label{pm-mgi2}
G^{1000}G^{0111}-G^{0100}G^{1011}=\pm(G^{0010}G^{1101}-G^{0001}G^{1110}),
\end{equation}
then $G^{-1}$ is of the following form
\begin{equation*}
G^{-1}=\begin{pmatrix}
0&g^{0001}&g^{0010}&0\\
g^{0100}&0&0&g^{0111}\\
g^{1000}&0&0&g^{1011}\\
0&g^{1101}&g^{1110}&0
\end{pmatrix}
\end{equation*}
and
\begin{equation*}
g^{1000}g^{0111}-g^{0100}g^{1011}=\pm(g^{0010}g^{1101}-g^{0001}g^{1110}),
\end{equation*}
with the same sign $\pm$ as in (\ref{pm-mgi2}).
\end{lemma}

\begin{remark}
In fact, Lemma \ref{even arity 4} and Lemma \ref{odd arity 4} are the group properties of standard signatures of arity 4, where $\pm$ is due to
the exact ordering of the 4 bits when we defined the matrix relative to
MGI.
\end{remark}



\begin{theorem}\label{gplemma4}
Let $\underline{G}=(\underline{G}^{\alpha_{1}\alpha_{2}\cdots \alpha_{n}})$ be the standard signature of a generator matchgate $\Gamma$ of arity $2n$. If $\underline{G}$ has full rank and the $t$-th matrix form $\underline{G}(t)$ has rank 4, then  there exists a standard signature $\underline{R}$ realized by a recognizer matchgate of arity $2n$  such that $\underline{G}(t)\underline{R}(t)=I_{4}$.
\end{theorem}
\begin{proof}
We follow the notations of Lemma \ref{lemma5.2} and Lemma \ref{lemma5.3}.
From Theorem \ref{sub-basis}, there is a sub-matrix of rank 4
\begin{equation*}A=\begin{pmatrix}
\underline{G}^{\alpha\curvearrowleft_{t}00}&\underline{G}^{\beta\curvearrowleft_{t}00}
&\underline{G}^{\gamma\curvearrowleft_{t}00}
&\underline{G}^{\delta\curvearrowleft_{t}00}\\
\underline{G}^{\alpha\curvearrowleft_{t}01}&\underline{G}^{\beta\curvearrowleft_{t}01}
&\underline{G}^{\gamma\curvearrowleft_{t}01}
&\underline{G}^{\delta\curvearrowleft_{t}01}\\
\underline{G}^{\alpha\curvearrowleft_{t}10}&\underline{G}^{\beta\curvearrowleft_{t}10}
&\underline{G}^{\gamma\curvearrowleft_{t}10}
&\underline{G}^{\delta\curvearrowleft_{t}10}\\
\underline{G}^{\alpha\curvearrowleft_{t}11}&\underline{G}^{\beta\curvearrowleft_{t}11}
&\underline{G}^{\gamma\curvearrowleft_{t}11}
&\underline{G}^{\delta\curvearrowleft_{t}11}\end{pmatrix}
\end{equation*}
in $\underline{G}(t)$, where the $q_{1}, q_{2}$-th bits of $\alpha, \beta, \gamma, \delta$ are 00, 01, 10, 11 respectively, and all other bits are the same.

By the Parity Condition, $A$ is of the form
\begin{equation}\label{matrix form for even}
A=\begin{pmatrix}
\underline{G}^{\alpha\curvearrowleft_{t}00}&0
&0&\underline{G}^{\delta\curvearrowleft_{t}00}\\
0&\underline{G}^{\beta\curvearrowleft_{t}01}
&\underline{G}^{\gamma\curvearrowleft_{t}01}&0\\
0&\underline{G}^{\beta\curvearrowleft_{t}10}
&\underline{G}^{\gamma\curvearrowleft_{t}10}&0\\
\underline{G}^{\alpha\curvearrowleft_{t}11}&0
&0&\underline{G}^{\delta\curvearrowleft_{t}11}
\end{pmatrix}
\end{equation}

or
\begin{equation}\label{matrix form for odd}
A=\begin{pmatrix}
0&\underline{G}^{\beta\curvearrowleft_{t}00}
&\underline{G}^{\gamma\curvearrowleft_{t}00}&0\\
\underline{G}^{\alpha\curvearrowleft_{t}01}&0
&0&\underline{G}^{\delta\curvearrowleft_{t}01}\\
\underline{G}^{\alpha\curvearrowleft_{t}10}&0
&0&\underline{G}^{\delta\curvearrowleft_{t}10}\\
0&\underline{G}^{\beta\curvearrowleft_{t}11}
&\underline{G}^{\gamma\curvearrowleft_{t}11}&0
\end{pmatrix}.
\end{equation}
We prove Theorem \ref{gplemma4} for the case in (\ref{matrix form for even}). The other case is similar.

Let the position vector be $(q_{1}q_{2})\curvearrowleft_{t}(p_{1}p_{2})$ and the pattern be $\alpha\curvearrowleft_{t}10$, then we have from MGI
\begin{equation*}
\underline{G}^{\alpha\curvearrowleft_{t}00}\underline{G}^{\delta\curvearrowleft_{t}11}-\underline{G}^{\alpha\curvearrowleft_{t}11}
\underline{G}^{\delta\curvearrowleft_{t}00}=\pm(\underline{G}^{\gamma\curvearrowleft_{t}10}\underline{G}^{\beta\curvearrowleft_{t}01}-
\underline{G}^{\beta\curvearrowleft_{t}10}\underline{G}^{\gamma\curvearrowleft_{t}01}).
\end{equation*}
It is $``+"$ if $q_{1}<p_{1}<p_{2}<q_{2}$. Otherwise it is $``-"$.

Let $\underline{R}$ be a vector of dimension $2^{2n}$, where
\begin{equation*}\begin{pmatrix}
\underline{R}_{\alpha\curvearrowleft_{t}00}&0
&0&\underline{R}_{\alpha\curvearrowleft_{t}11}\\
0&\underline{R}_{\beta\curvearrowleft_{t}01}
&\underline{R}_{\beta\curvearrowleft_{t}10}&0\\
0&\underline{R}_{\gamma\curvearrowleft_{t}01}
&\underline{R}_{\gamma\curvearrowleft_{t}10}&0\\
\underline{R}_{\delta\curvearrowleft_{t}00}&0
&0&\underline{R}_{\delta\curvearrowleft_{t}11}
\end{pmatrix}
=A^{-1},
\end{equation*}
 and all other entries of $\underline{R}$ are zero.
It is obvious that $\underline{G}(t)\underline{R}(t)=I_{4}$.
Furthermore, by Lemma \ref{even arity 4}, $\underline{R}$ satisfies
\begin{equation*}
\underline{R}_{\alpha\curvearrowleft_{t}00}\underline{R}_{\delta\curvearrowleft_{t}11}-\underline{R}_{\alpha\curvearrowleft_{t}11}
\underline{R}_{\delta\curvearrowleft_{t}00}=\pm(\underline{R}_{\gamma\curvearrowleft_{t}10}\underline{R}_{\beta\curvearrowleft_{t}01}-
\underline{R}_{\beta\curvearrowleft_{t}10}\underline{R}_{\gamma\curvearrowleft_{t}01}).
\end{equation*}
It is $``+"$ if $q_{1}<p_{1}<p_{2}<q_{2}$. Otherwise is $``-"$.
Note that this is  the only non-trivial matchgate identity for $\underline{R}$. Thus $\underline{R}$ is a standard signature realized by a recognizer matchgate by Lemma \ref{MGI}.

\end{proof}


\subsection{A Collapse Theorem on Domain Size 4}
In this subsection, assume that $G$ is a generator signature of full rank on domain size 4,
the basis $M$ is a $2^{\ell}\times 4$ matrix of rank $4$, and $\underline{G}=M^{\otimes n}G$ is a standard signature of arity $n\ell$.
Since $G$ is a signature of full rank on domain size 4,
 there exists $t\in[n]$ such that rank$(G(t))=4$. Following the notations of Lemma \ref{lemma5.2} and Lemma \ref{lemma5.3},
denote $\{\sigma, \sigma\oplus e_{{p}_{1}}, \sigma\oplus e_{{p}_{2}},
\sigma\oplus e_{{p}_{1}}\oplus e_{{p}_{2}}\}$
as $\sigma+\{e_{{p}_{1}}, e_{{p}_{2}}\}$.
Note that $\sigma\oplus e_{{p}_{1}}\oplus e_{{p}_{2}}=\tau$.

\begin{lemma}
$M^{\sigma+\{e_{{p}_{1}}, e_{{p}_{2}}\}}$ is invertible.
\end{lemma}
\begin{proof}
Note that the sub-matrix of rank 4
in the proof of Theorem \ref{sub-basis}
 is a sub-matrix of $M^{\sigma+\{e_{{p}_{1}}, e_{{p}_{2}}\}}$
$G(t)$$(M^{\tt{T}})^{\otimes (n-1)}$,
 so $M^{\sigma+\{e_{{p}_{1}}, e_{{p}_{2}}\}}$
has rank at least 4. Furthermore, note that $M^{\sigma+\{e_{{p}_{1}}, e_{{p}_{2}}\}}$ is a $4\times 4$ matrix, so $M^{\sigma+\{e_{{p}_{1}}, e_{{p}_{2}}\}}$ is invertible.
\end{proof}

Note that $(M^{\sigma+\{e_{{p}_{1}}, e_{{p}_{2}}\}})^{\otimes n}G$ is a column vector of dimension $2^{2n}$ and we denote it by
$\underline{G}^{*\leftarrow \sigma+\{e_{{p}_{1}}, e_{{p}_{2}}\}}$.

\begin{theorem}
\label{lemma4.3}
$\underline{G}^{*\leftarrow \sigma+\{e_{{p}_{1}}, e_{{p}_{2}}\}}=(M^{\sigma+\{e_{{p}_{1}}, e_{{p}_{2}}\}})^{\otimes n}G$ is the standard signature of a generator matchgate of arity $2n$
 and $\underline{G}^{*\leftarrow\sigma+\{e_{{p}_{1}}, e_{{p}_{2}}\}}(t)$ has rank 4.
\end{theorem}

\begin{proof}
Let $\Gamma$ be a matchgate realizing the standard signature $\underline{G}=M^{\otimes n}G$. Note that $\Gamma$ has $n\ell$ external nodes.
For every block of $\ell$ nodes, we add an edge of weight 1 to the $i$-th external node if the $i$-th bit of $\alpha$ is 1 and do nothing to it if the $i$-th bit of $\alpha$ is 0 for $1\leq i\leq \ell$. Then we get a new matchgate $\Gamma'$. Furthermore, for every block of $\Gamma'$, view the ${p}_{1}, {p}_{2}$-th external nodes as external nodes and all others as internal nodes, we get a matchgate realizing $\underline{G}^{*\leftarrow\sigma+\{e_{{p}_{1}}, e_{{p}_{2}}\}}=(M^{\sigma+\{e_{{p}_{1}}, e_{{p}_{2}}\}})^{\otimes n}G$.
Since $M^{\sigma+\{e_{{p}_{1}}, e_{{p}_{2}}\}}$ has rank 4, $\underline{G}^{*\leftarrow\sigma+\{e_{{p}_{1}}, e_{{p}_{2}}\}}(t)=M^{\sigma+\{e_{{p}_{1}}, e_{{p}_{2}}\}}G(t)((M^{\sigma+\{e_{{p}_{1}}, e_{{p}_{2}}\}})^{\tt{T}})^{\otimes (n-1)}$ has rank 4.
\end{proof}

Note that $(M^{\sigma+\{e_{{p}_{1}}, e_{{p}_{2}}\}})^{\otimes(t-1)}\otimes M\otimes (M^{\sigma+\{e_{{p}_{1}}, e_{{p}_{2}}\}})^{\otimes(n-t)}
\cdot G$ is a column vector of dimension $2^{2n+\ell-2}$ and denote it by $\underline{G}^{t^{c}\leftarrow\sigma+\{e_{{p}_{1}}, e_{{p}_{2}}\}}$.

\begin{lemma}\label{lemma5.7}
$\underline{G}^{t^{c}\leftarrow\sigma+\{e_{{p}_{1}}, e_{{p}_{2}}\}}$=$(M^{\sigma+\{e_{{p}_{1}}, e_{{p}_{2}}\}})^{\otimes(t-1)}\otimes M\otimes (M^{\sigma+\{e_{{p}_{1}}, e_{{p}_{2}}\}})^{\otimes(n-t)}G$ is the standard signature
of a generator matchgate of arity $2n+\ell-2$.
\end{lemma}
\begin{proof}
Let $\Gamma$ be a matchgate realizing the standard signature $\underline{G}=M^{\otimes n}G$. Note that $\Gamma$ has $n\ell$ external nodes. We do nothing to the $t$-th block.
For other blocks, we add an edge of weight 1 to the $i$-th external node if the $i$-th bit of $\alpha$ is 1 and do nothing to it if the $i$-th bit of $\alpha$ is 0 for $1\leq i\leq \ell$. Then we get a new matchgate $\Gamma'$. The external nodes of $\Gamma'$ consists of the external nodes in
the $t$-th block, and the ${p}_{1}, {p}_{2}$-th external nodes in the other blocks, then we get a matchgate realizing $\underline{G}^{t^{c}\leftarrow\sigma+\{e_{{p}_{1}}, e_{{p}_{2}}\}}$.
\end{proof}

Let
$T=M(M^{\sigma+\{e_{{p}_{1}}, e_{{p}_{2}}\}})^{-1}$ (Note that $T^{\sigma+\{e_{{p}_{1}}, e_{{p}_{2}}\}}=I_{4}$), then
\begin{equation*}
\underline{G}=M^{\otimes n}G=T^{\otimes n}(M^{\sigma+\{e_{{p}_{1}}, e_{{p}_{2}}\}})^{\otimes n}G=T^{\otimes n}\underline{G}^{*\leftarrow\sigma+\{e_{{p}_{1}}, e_{{p}_{2}}\}}
\end{equation*}
and
\begin{equation*}
\underline{G}^{t^{c}\leftarrow\sigma+\{e_{{p}_{1}}, e_{{p}_{2}}\}}=(T^{\sigma+\{e_{{p}_{1}}, e_{{p}_{2}}\}})^{\otimes(t-1)}\otimes T\otimes (T^{\sigma+\{e_{{p}_{1}}, e_{{p}_{2}}\}})^{\otimes(n-t)}\underline{G}^{*\leftarrow\sigma+\{e_{{p}_{1}}, e_{{p}_{2}}\}}.
\end{equation*}

The entries of $\underline{G}^{t^{c}\leftarrow\sigma+\{e_{{p}_{1}}, e_{{p}_{2}}\}}$ can be indexed by
$i_{1,1}i_{1,2}\cdots i_{t-1,1}i_{t-1,2}i'_{1}\cdots i'_{\ell}i_{t+1,1}i_{t+1,2}\cdots$
$i_{n,1}i_{n,2}\in\{0,1\}^{2n+\ell-2}$.
Denote the matrix form of $\underline{G}^{t^{c}\leftarrow\sigma+\{e_{{p}_{1}}, e_{{p}_{2}}\}}$ whose rows are indexed by
$i'_{1}\cdots i'_{\ell}$ and columns indexed by $i_{1,1}i_{1,2}\cdots i_{t-1,1}i_{t-1,2}i_{t+1,1}i_{t+1,2}\cdots i_{n,1}i_{n,2}$
as $\underline{G}^{t^{c}\leftarrow\sigma+\{e_{{p}_{1}}, e_{{p}_{2}}\}}(t)$, then
\begin{equation}\underline{G}^{t^{c}\leftarrow\sigma+\{e_{{p}_{1}}, e_{{p}_{2}}\}}(t)=
 T\underline{G}^{*\leftarrow\sigma+\{e_{{p}_{1}}, e_{{p}_{2}}\}}(t)((T^{\sigma+\{e_{{p}_{1}}, e_{{p}_{2}}\}})^{\tt{T}})^{\otimes (n-1)}
 =T\underline{G}^{*\leftarrow\sigma+\{e_{{p}_{1}}, e_{{p}_{2}}\}}(t).
 \end{equation}

\begin{lemma}
\label{lemma4.5}
$T$ is the standard signature of a transducer matchgate
of $\ell$-output and 2-input.
\end{lemma}
\begin{proof}
By Theorem \ref{gplemma4} and Theorem \ref{lemma4.3}, there exists a standard signature $\underline{R}$ realized by a recognizer matchgate of arity $2n$
  such that $\underline{G}^{*\leftarrow\sigma+\{e_{{p}_{1}}, e_{{p}_{2}}\}}(t)\underline{R}(t)=I_{4}$ (Note that $\underline{R}(t)$ is a $2^{2n-2}\times 4$ matrix).
 Let $\Gamma_{1}$ be the matchgate realizing $\underline{G}^{t^{c}\leftarrow\sigma+\{e_{{p}_{1}}, e_{{p}_{2}}\}}$ with output nodes
 $X_{1,1}, X_{1,2}, \cdots, X_{t-1,1}, X_{t-1, 2}, Y_{1}, Y_{2}, \cdots, Y_{\ell}$, $Z_{t+1, 1}, Z_{t+1, 2}, \cdots, Z_{n, 1}, Z_{n, 2}$, and $\Gamma_{2}$ be
 the matchgate realizing $\underline{R}$ with input nodes $W_{1, 1}, W_{1, 2}, W_{2, 1}, W_{2, 2}$, $\cdots, W_{n, 1}, W_{n, 2}$.
Connect $X_{i, 1}$ with $W_{i, 1}$, $X_{i, 2}$ with $W_{i, 2}$, for $1\leq i\leq t-1$ and
$Z_{i, 1}$ with $W_{i, 1}$, $Z_{i, 2}$ with $W_{i, 2}$, for $t+1\leq i\leq n$
 by an edge with weight 1 respectively, then we get a transducer matchgate $\Gamma$ with output nodes $Y_{1}, Y_{2}, \cdots, Y_{\ell}$ and input nodes $W_{t, 1}$, $W_{t, 2}$.
And $T=\underline{G}^{t^{c}\leftarrow\sigma, \tau}(t)\underline{R}(t)$ is the standard signature of $\Gamma$.
\end{proof}
The proof of Lemma \ref{lemma4.5} is illustrated by Fig. 2.

\begin{theorem}\label{collapse2}
Any holographic algorithm on a basis of size $\ell$ and domain size 4 which employs at least one generator signature of full rank can be simulated
on a basis of size 2.
\end{theorem}
\begin{proof}
Let $\underline{R}_{i}M^{\otimes m_i}=R_{i}$ for $1\leq i\leq r$ and $\underline{G}_{j}=M^{\otimes n_j}G_{j}$ for $1\leq j\leq g$, where $R_{i}, G_{j}$
 are recognizer and generator signatures that a holographic algorithm employs and $\underline{R}_{i}, \underline{G}_{j}$ are standard signatures.
Without loss of generality, let $G_{1}$ be of full rank, and $\underline{G}_{1}=M^{\otimes n_1}G_{1}$. Starting from $G_{1}$ we define
$\sigma$ and $\tau$.
 Then the basis $M$ has a full rank sub-matrix $M^{\sigma+\{e_{{p}_{1}}, e_{{p}_{2}}\}}$ and $T=M(M^{\sigma+\{e_{{p}_{1}}, e_{{p}_{2}}\}})^{-1}$ is the standard signature of a transducer matchgate. Let $\underline{R}_{i}'=\underline{R}_{i}T^{\otimes m_i}$,
 then
 \begin{equation*}
 \underline{R}'_{i}(M^{\sigma+\{e_{{p}_{1}}, e_{{p}_{2}}\}})^{\otimes m_i}=R_{i},~~~~
 \underline{G}_{j}^{*\leftarrow \sigma+\{e_{{p}_{1}}, e_{{p}_{2}}\}}=(M^{\sigma+\{e_{{p}_{1}}, e_{{p}_{2}}\}})^{\otimes n_j}G_{j},
 \end{equation*}
  for $1\leq i\leq r$, $1\leq j\leq g$, where $\underline{R}_{i}'$
and $\underline{G}_{j}^{*\leftarrow\sigma+\{e_{{p}_{1}}, e_{{p}_{2}}\}} $  are standard signatures by Lemma \ref{product transformer} and Theorem \ref{lemma4.3}. This implies that $R_{i}, G_{j}$ are simultaneously realized on the basis $M^{\sigma+\{e_{{p}_{1}}, e_{{p}_{2}}\}}$.
\end{proof}

\section{Acknowledgments}
We would like to thank Pinyan Lu, Tyson Williams, Heng Guo, Leslie Valiant for their interest and especially
Tyson Williams for his computer code in our experimentation.

\renewcommand{\refname}{References}

\newpage
\section*{Figures}
\setlength{\unitlength}{5mm}
\begin{picture}(30,43)(-2,0)
\put(0,0){\line(0,1){41}}
\put(0,0){\line(1,0){14}}
\put(14,0){\line(0,1){41}}
\put(0,41){\line(1,0){14}}
\put(16,0){\line(0,1){41}}
\put(16,0){\line(1,0){14}}
\put(30,0){\line(0,1){41}}
\put(16,41){\line(1,0){14}}
\put(9,2){\line(0,1){6}}
\put(9,2){\line(1,0){4}}
\put(13,2){\line(0,1){6}}
\put(9,8){\line(1,0){4}}
\put(8,5){$X_{t-1}$}
\put(11.5,5){$\bullet$}
\put(11.5,5.2){\line(1,0){6.5}}
\put(11,10.5){\vdots}
\put(9,14){\line(0,1){6}}
\put(9,14){\line(1,0){4}}
\put(13,14){\line(0,1){6}}
\put(9,20){\line(1,0){4}}
\put(8,17){$X_{1}$}
\put(11.5,17){$\bullet$}
\put(11.5,17.2){\line(1,0){6.5}}
\put(9,21){\line(0,1){6}}
\put(9,21){\line(1,0){4}}
\put(13,21){\line(0,1){6}}
\put(9,27){\line(1,0){4}}
\put(8,24){$Z_{n}$}
\put(11.5,24){$\bullet$}
\put(11.5,24.2){\line(1,0){6.5}}
\put(11,29.5){\vdots}
\put(9,33){\line(0,1){6}}
\put(9,33){\line(1,0){4}}
\put(13,33){\line(0,1){6}}
\put(9,39){\line(1,0){4}}
\put(8,36){$Z_{t+1}$}
\put(11.5,36){$\bullet$}
\put(11.5,36.2){\line(1,0){6.5}}
\put(1,16.5){\line(0,1){6}}
\put(1,16.5){\line(1,0){4}}
\put(5,16.5){\line(0,1){6}}
\put(1,22.5){\line(1,0){4}}
\put(2.5,17.5){$Y_{1}$}
\put(2.5,19.5){$Y_{2}$}
\put(2.5,21.5){$Y_{\ell}$}
\put(1.5,18){$\bullet$}
\put(1.5,19.5){$\bullet$}
\put(1.5,21.2){$\bullet$}
\put(1.6,20.1){\vdots}
\put(5,34){$\Gamma_{1}$}
\put(17,2){\line(0,1){6}}
\put(17,2){\line(1,0){4}}
\put(21,2){\line(0,1){6}}
\put(17,8){\line(1,0){4}}
\put(18,5){$\bullet$}
\put(19,5){$W_{t-1}$}
\put(19,10.5){\vdots}
\put(17,14){\line(0,1){6}}
\put(17,14){\line(1,0){4}}
\put(21,14){\line(0,1){6}}
\put(17,20){\line(1,0){4}}
\put(18,17){$\bullet$}
\put(19,17){$W_{1}$}
\put(17,21){\line(0,1){6}}
\put(17,21){\line(1,0){4}}
\put(21,21){\line(0,1){6}}
\put(17,27){\line(1,0){4}}
\put(18,24){$\bullet$}
\put(19,24){$W_{n}$}
\put(19,29.5){\vdots}
\put(17,33){\line(0,1){6}}
\put(17,33){\line(1,0){4}}
\put(21,33){\line(0,1){6}}
\put(17,39){\line(1,0){4}}
\put(18,36){$\bullet$}
\put(19,36){$W_{t+1}$}
\put(25,16.5){\line(0,1){6}}
\put(25,16.5){\line(1,0){4}}
\put(29,16.5){\line(0,1){6}}
\put(25,22.5){\line(1,0){4}}
\put(28,20){$\bullet$}
\put(26.5,20){$W_{t}$}
\put(28,34){$\Gamma_{2}$}
\put(14.2, -2){\scriptsize Fig. 1}
\end{picture}
\newpage
\setlength{\unitlength}{5mm}
\begin{picture}(30,43)(-2,0)

\put(0,0){\line(0,1){41}}
\put(0,0){\line(1,0){14}}
\put(14,0){\line(0,1){41}}
\put(0,41){\line(1,0){14}}
\put(16,0){\line(0,1){41}}
\put(16,0){\line(1,0){14}}
\put(30,0){\line(0,1){41}}
\put(16,41){\line(1,0){14}}
\put(9,2){\line(0,1){6}}
\put(9,2){\line(1,0){4}}
\put(13,2){\line(0,1){6}}
\put(9,8){\line(1,0){4}}
\put(11.5,4){$\bullet$}
\put(11.5,6){$\bullet$}
\put(8,4){$X_{t-1,2}$}
\put(8,6){$X_{t-1,1}$}
\put(11.6,3){\vdots}
\put(11.6,5){\vdots}
\put(11.6,7){\vdots}
\put(11.5,4.2){\line(1,0){6.5}}
\put(11.5,6.2){\line(1,0){6.5}}
\put(11,10.5){\vdots}
\put(9,14){\line(0,1){6}}
\put(9,14){\line(1,0){4}}
\put(13,14){\line(0,1){6}}
\put(9,20){\line(1,0){4}}
\put(11.5,16){$\bullet$}
\put(11.5,18){$\bullet$}
\put(8,16){$X_{1,2}$}
\put(8,18){$X_{1,1}$}
\put(11.6,15){\vdots}
\put(11.6,17){\vdots}
\put(11.6,19){\vdots}
\put(11.5,16.2){\line(1,0){6.5}}
\put(11.5,18.2){\line(1,0){6.5}}
\put(9,21){\line(0,1){6}}
\put(9,21){\line(1,0){4}}
\put(13,21){\line(0,1){6}}
\put(9,27){\line(1,0){4}}
\put(11.5,23){$\bullet$}
\put(11.5,25){$\bullet$}
\put(8,23){$Z_{n,2}$}
\put(8,25){$Z_{n,1}$}
\put(11.6,22){\vdots}
\put(11.6,24){\vdots}
\put(11.6,26){\vdots}
\put(11.5,23.2){\line(1,0){6.5}}
\put(11.5,25.2){\line(1,0){6.5}}
\put(11,29.5){\vdots}
\put(9,33){\line(0,1){6}}
\put(9,33){\line(1,0){4}}
\put(13,33){\line(0,1){6}}
\put(9,39){\line(1,0){4}}
\put(11.5,35){$\bullet$}
\put(11.5,37){$\bullet$}
\put(8,35){$Z_{t+1,2}$}
\put(8,37){$Z_{t+1,1}$}
\put(11.6,34){\vdots}
\put(11.6,36){\vdots}
\put(11.6,38){\vdots}
\put(11.5,35.2){\line(1,0){6.5}}
\put(11.5,37.2){\line(1,0){6.5}}
\put(1,16.5){\line(0,1){6}}
\put(1,16.5){\line(1,0){4}}
\put(5,16.5){\line(0,1){6}}
\put(1,22.5){\line(1,0){4}}
\put(2.5,17.5){$Y_{1}$}
\put(2.5,19.5){$Y_{2}$}
\put(2.5,21.5){$Y_{\ell}$}
\put(1.5,18){$\bullet$}
\put(1.5,19.5){$\bullet$}
\put(1.5,21.2){$\bullet$}
\put(1.6,20.1){\vdots}
\put(5,34){$\Gamma_{1}$}
\put(17,2){\line(0,1){6}}
\put(17,2){\line(1,0){4}}
\put(21,2){\line(0,1){6}}
\put(17,8){\line(1,0){4}}
\put(18,4){$\bullet$}
\put(18,6){$\bullet$}
\put(19,4){$W_{t-1,2}$}
\put(19,6){$W_{t-1,1}$}
\put(19,10.5){\vdots}
\put(17,14){\line(0,1){6}}
\put(17,14){\line(1,0){4}}
\put(21,14){\line(0,1){6}}
\put(17,20){\line(1,0){4}}
\put(18,16){$\bullet$}
\put(18,18){$\bullet$}
\put(19,16){$W_{1,2}$}
\put(19,18){$W_{1,1}$}
\put(17,21){\line(0,1){6}}
\put(17,21){\line(1,0){4}}
\put(21,21){\line(0,1){6}}
\put(17,27){\line(1,0){4}}
\put(18,23){$\bullet$}
\put(18,25){$\bullet$}
\put(19,23){$W_{n,2}$}
\put(19,25){$W_{n,1}$}
\put(19,29.5){\vdots}
\put(17,33){\line(0,1){6}}
\put(17,33){\line(1,0){4}}
\put(21,33){\line(0,1){6}}
\put(17,39){\line(1,0){4}}
\put(18,35){$\bullet$}
\put(18,37){$\bullet$}
\put(19,35){$W_{t+1,2}$}
\put(19,37){$W_{t+1,1}$}
\put(25,16.5){\line(0,1){6}}
\put(25,16.5){\line(1,0){4}}
\put(29,16.5){\line(0,1){6}}
\put(25,22.5){\line(1,0){4}}
\put(28,19){$\bullet$}
\put(25.5,19){$W_{t,1}$}
\put(28,21){$\bullet$}
\put(25.5,21){$W_{t,2}$}
\put(28,34){$\Gamma_{2}$}
\put(14.2, -2){\scriptsize Fig. 2}
\end{picture}

\section*{Appendix}

We consider the following problem on domain size 4
as an illustration of problems solved by holographic algorithms
using matchgates.

\vspace{.1in}
\noindent
\textbf{Doppler Shift Problem}

\noindent
\textbf{Input:} An undirected 3-regular graph $G$.

\noindent
\textbf{Output:} The number of $\{{\rm Red, Yellow, Green, Blue}\}$-colorings
of the edges of $G$, such that at every vertex,
either there is a {\it Red}-shift, namely all incident edges
are colored with Red or Yellow, or Green, but not Blue,
or there is a {\it Blue}-shift, namely all incident edges
are colored with Blue, or Green, or Yellow, but not Red.

We consider the following basis $M \in \mathbb{C}^{4 \times 4}$, where
\[ M = \begin{pmatrix}
1 & 1 & 1 & 1\\
1 & 1 & -1 & -1\\
1 & -1 & 1 & -1\\
1 & -1 & -1 & 1
\end{pmatrix}.\]
If we index the rows and columns of $M$ by $x_1 x_2 \in \{0, 1\}^2$
and $y_1 y_2 \in \{0, 1\}^2$
respectively, then the entry indexed by $(x_1 x_2, y_1 y_2)$
is $(-1)^{x_1 y_2 + x_2 y_1}$.
If we identify $00, 01, 10, 11$ with the colors
Red, Yellow, Green, Blue, respectively,
then it can be directly verified that
an arity 3 function on domain $\{0, 1\}^2$
representing the local constraint of this problem
under the holographic transformation of $M^{\otimes 3}$
can be realized by a matchgate of arity 6.
As this $M$ is  (a scalar multiple of) an orthogonal matrix,
we see that this prbolem can be solved by matchgates
in polynomial time after the holographic transformation.

\end{document}